\documentclass[11pt]{article}
\usepackage[margin = 1in]{geometry}
\usepackage{xcolor}
\usepackage{fancyhdr}
\usepackage{lastpage}
\usepackage{parskip} 
\usepackage{amsmath} 
\usepackage{amsthm} 
\usepackage{amssymb}
\usepackage{graphicx}
\usepackage{proof}
\usepackage{enumitem}
\usepackage{algorithm}
\usepackage{algorithmicx}
\usepackage{algpseudocode}
\usepackage{hyperref}
\usepackage{subcaption}
 \usepackage{setspace}
 \usepackage{url}
 \usepackage[square,sort,comma,numbers]{natbib}
 \usepackage{apalike}
 \usepackage{helvet}

\usepackage{natbib}
\bibpunct{[}{]}{,}{a}{}{;}
\def\newblock{\ }%

\newtheorem{theorem}{Theorem}
\newtheorem{corollary}{Corollary}
\newtheorem{lemma}{Lemma}
\newtheorem{proposition}{Proposition}
\newtheorem{remark}{Remark}
\newtheorem{example}{Example}

\renewcommand{\algorithmicrequire}{\textbf{Input:}}  

\title{Cascading Losses in Reinsurance Networks
\footnote{We thank Steffen Schuldenzucker for his valuable contribution in the proof to Thereom~\ref{prop:net_liabilities_equal} and Sven Seuken and the Economics \& Computation group at University of Zurich for helpful discussions. We also thank Frank Krieter, Dominic Rau, and the risk team at Swiss Re and Sean Bourgeois at Tremor for valuable discussions covering the reinsurance industry. The first author acknowledges an Amherst College Fellowship. This work was funded through NSF RTG Award  \#1645643, NSF CRISP Award \#1638230 and  NSF CAREER Award \#1653354.
	}
}

\author{Ariah Klages-Mundt\thanks{Cornell University, Center for Applied Mathematics, Ithaca, NY 14850, USA, email: {\tt aak228@cornell.edu}.}   \ \ \ \ \ \ \ \
Andreea Minca\thanks{Cornell University, School of Operations Research and Information Engineering, Ithaca, NY, 14850, USA, email: {\tt acm299@cornell.edu}.}
}

\date{May 1, 2019}

\begin{document}

\maketitle

\begin{abstract}
\noindent We develop a model for contagion in reinsurance networks by which primary insurers' losses are spread through the network. Our model handles general reinsurance contracts, such as typical excess of loss contracts. We show that simpler models existing in the literature--namely proportional reinsurance--greatly underestimate contagion risk. We characterize the fixed points of our model and develop efficient algorithms to compute contagion with guarantees on convergence and speed under conditions on network structure. We characterize exotic cases of problematic graph structure and nonlinearities, which cause network effects to dominate the overall payments in the system. We lastly apply our model to data on real world reinsurance networks. Our simulations demonstrate the following:
\begin{itemize}
\item Reinsurance networks face extreme sensitivity to parameters. A firm can be wildly uncertain about its losses even under small network uncertainty.
\item Our sensitivity results reveal a new incentive for firms to cooperate to prevent fraud, as even small cases of fraud can have outsized effect on the losses across the network.
\item Nonlinearities from excess of loss contracts obfuscate risks and can cause excess costs in a real world system.
\end{itemize}
\end{abstract}

\section{Introduction}

The London market excess of loss (LMX) spirals of the 1980-90s revealed how global interconnections among reinsurers (i.e., insurers who insure other insurers) can cause contagion in the reinsurance market \citep{bain99}. There was high concentration of losses despite the belief that all parties were properly insured. A series of major storms caused tail losses to the London insurance market (Lloyd's in particular). While risks in the London market were reinsured outside the UK, retrocession (i.e., reinsurance on reinsurance) brought these losses back to the London market, resulting in unexpected concentration of losses. Figure~\ref{fig:LMX_diagram} visualizes these interconnections.

\begin{figure}
	\centering \includegraphics[width=11cm]{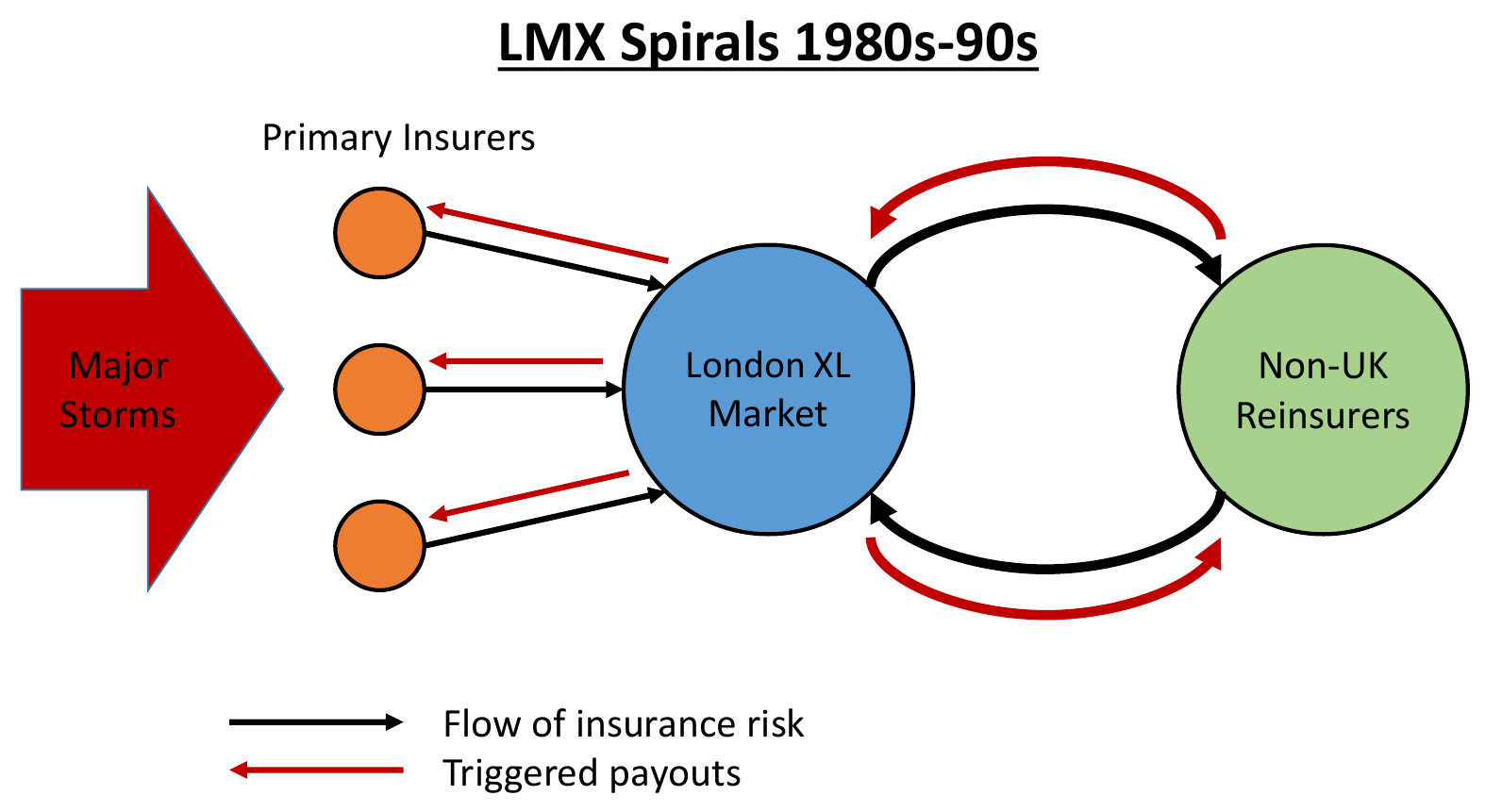}
	\caption{Diagram of LMX reinsurance spirals.}\label{fig:LMX_diagram}
\end{figure}

After these events, the industry mitigated spiral risks by reducing the size of the retrocession market. Today, there is a sense in insurance that the risk of spirals is largely a thing of the past and that risks are properly shared with reinsurers. To our knowledge, no reinsurance risk models used in industry directly account for these network effects. By applying the machinery we develop in this paper to estimates of the current US reinsurance system, we show that the reinsurance market is, in fact, not safe to network effects. We show such network effects can dominate the tail behavior of the system in ways that are difficult to predict. The US has insurance guaranty mechanisms that protect policyholders in case of insurance company insolvency. Our results are even more relevant in this case because the spiraling losses would be borne by the state.

We propose a model for contagion in reinsurance markets by which primary insurers' losses are spread throughout the network. Despite a vast literature on contagion in financial networks (see e.g., \citep{eisenberg01}, \citep{acemoglu15}, \citep{jackson14}), no existing contagion models are general enough to cover reinsurance contracts. The majority of financial network models are limited to simple contexts in which network interactions are representable by debt or equity contracts between entities. Very little work has extended these models to more complicated derivatives whose payoffs in equilibrium depend on the liabilities and counterparty risk across the network. As a notable exception, \citep{seuken16} and \citep{seuken17} demonstrate difficulties in clearing networks with credit default swaps in addition to initial debt contracts. Reinsurance contracts differ from debt contracts in that we do not outright know their liabilities. Reinsurance contracts differ from credit default swaps in that contract liabilities are not related to default events. Instead, the liabilities of reinsurance contracts are interrelated and nonlinear, which can lead to difficulties in determining equilibrium payoffs and to multiple solutions. Our model ventures far beyond the settings and results of \citep{eisenberg01} and \citep{acemoglu15} as it, in general, requires working with matrices with arbitrarily large column sums as opposed to column sums $\leq 1$.

\citep{blanchet15} developed one of the first network models for reinsurance contagion; however, they assume that reinsurance contracts are proportional contracts as opposed to the more common excess of loss contracts (we provide more background on these types of contracts in the next section) and that reinsurance contracts do not cover liabilities from other reinsurance contracts, which limits the propagation of losses to two steps in the network. These assumptions remove exotic behavior from the system, such as reinsurance spirals, which we show can play a critical role. Under these assumptions, they provide large deviation results for the loss in the system. In contrast, we focus on a more general setting that handles a wide variety of reinsurance contracts that exist in the real world, including the more common excess of loss contracts. Our simulations comparing excess of loss networks with proportional networks further show that this assumption in \citep{blanchet15} dangerously underestimates contagion risk in real reinsurance networks.

\citep{feinstein18} develops a dynamic framework for contingent claims that can accommodate some reinsurance contracts. However, the reinsurance contracts in their model cannot have caps. \citep{kley16} develop a bipartite graph model of tail risk in insurance. However, their model does not include reinsurance. 

\citep{feinstein17} describes the sensitivity of payment equilibria in \citep{eisenberg01} to small variations in the interbank liabilities. In contrast, our focus is on the reinsurance model that produces these liabilities. Further, we show that these liabilities can have wild variations from small uncertainties in network parameters.

In addition to developing a contagion model for reinsurance networks, our contributions include the following:
\begin{itemize}
\item We establish efficient algorithms to compute contagion with guarantees on convergence and speed under conditions on network structure.

\item We characterize exotic cases of problematic graph structure and nonlinearities, which cause network effects that dominate the overall payments in the system. We relate reinsurance spirals to structural properties of the network, such as the existence of graph cycles that recirculate large proportions of reinsurance losses. Further, we show that these cycles can be very complicated interactions of simple graph cycles.

\item We apply our model to real world reinsurance networks using data provided by the National Association of Insurance Commissioners (NAIC). Our simulations show that, using real world data, nonlinearities in contagion can cause extreme uncertainties. We demonstrate that even if a firm has unreasonably precise information\footnote{In the extreme, some real contracts are ambiguous to the degree that the parties to the contract themselves do not even know the contract parameters. We will discuss this further later in the paper.} (i.e., with small uncertainty) about the global structure of the system, it can still be wildly uncertain about the losses it will face from a given shock. We further demonstrate that these nonlinearities can cause excess costs in a real world system--i.e., the insurance-reinsurance system could be structured differently to perform its function to protect real world infrastructure more efficiently.
\end{itemize}
We conclude by introducing three promising starting points for solving real world issues that our results reveal: using distributed systems to control fraud, using network features to predict risk exposure, and designing markets to lower systemic costs.


\section{Reinsurance Contagion Model}

\subsection{Primer on reinsurance contracts}

Reinsurance contracts are insurance contracts that insurance companies take out to protect against large losses on their insurance portfolios. In \textbf{primary reinsurance}, the insurance company protected by the reinsurance is a primary insurance company. In \textbf{retrocession reinsurance}, the insurance company protected by the reinsurance is another reinsurance company. These reinsurance contracts are typically partly collateralized, meaning that, in the event that the reinsurer defaults on their obligations, the reinsured firm still has recourse to the collateral. Most reinsurance contracts in property and casualty are treaty contracts, which insure against losses from the reinsured company's entire insurance portfolio. Alternatively, in a niche case that we will not consider, some contracts have more specialized coverage of facultative risks.

The most common form of treaty reinsurance is an \textbf{excess of loss (XL)} contract, in which the reinsurer covers losses on the reinsured above a deductible (or attachment point). These contracts also commonly have caps (or limits) on the payouts of the contract. The total coverage of a firm is typically split into multiple deductible-cap layers in a tranche structure. Multiple reinsurers typically split each layer, taking fractions of the coverage. Together, the layers form a tower.

Another treaty contract is \textbf{proportional} reinsurance. These have no deductibles or caps, and the reinsurer takes on a percentage of the liabilities of the reinsured according to a coinsurance rate.

\subsection{Two contagion mechanisms}

Reinsurance contracts between a set of insurance companies form a network. Exogenous liabilities to a subset of primary insurers constitutes a shock to this network. This shock may activate the reinsurance to the primary insurers, which can in turn activate a cascade of retrocession reinsurance. Figure~\ref{fig:diagram_liability_propagation} outlines this \textbf{liability propagation} mechanism. The equilibrium of this process gives a network of liabilities between firms. Given their available capital, some firms may be unable to pay these liabilities. These firms default, potentially with extra default costs representing the legal, transactional, and liquidity costs of default. Each default negatively affects the capital of neighboring firms as these firms receive less on the liabilities they are owed. This can trigger a secondary cascade of defaults. Figure~\ref{fig:diagram_default_propagation} outlines this \textbf{default propagation} mechanism. An equilibrium of this second process is a clearing payment vector to the liability network.

\begin{figure}
	\centering
	\begin{subfigure}[b]{0.48\textwidth}
	 \includegraphics[width=\textwidth]{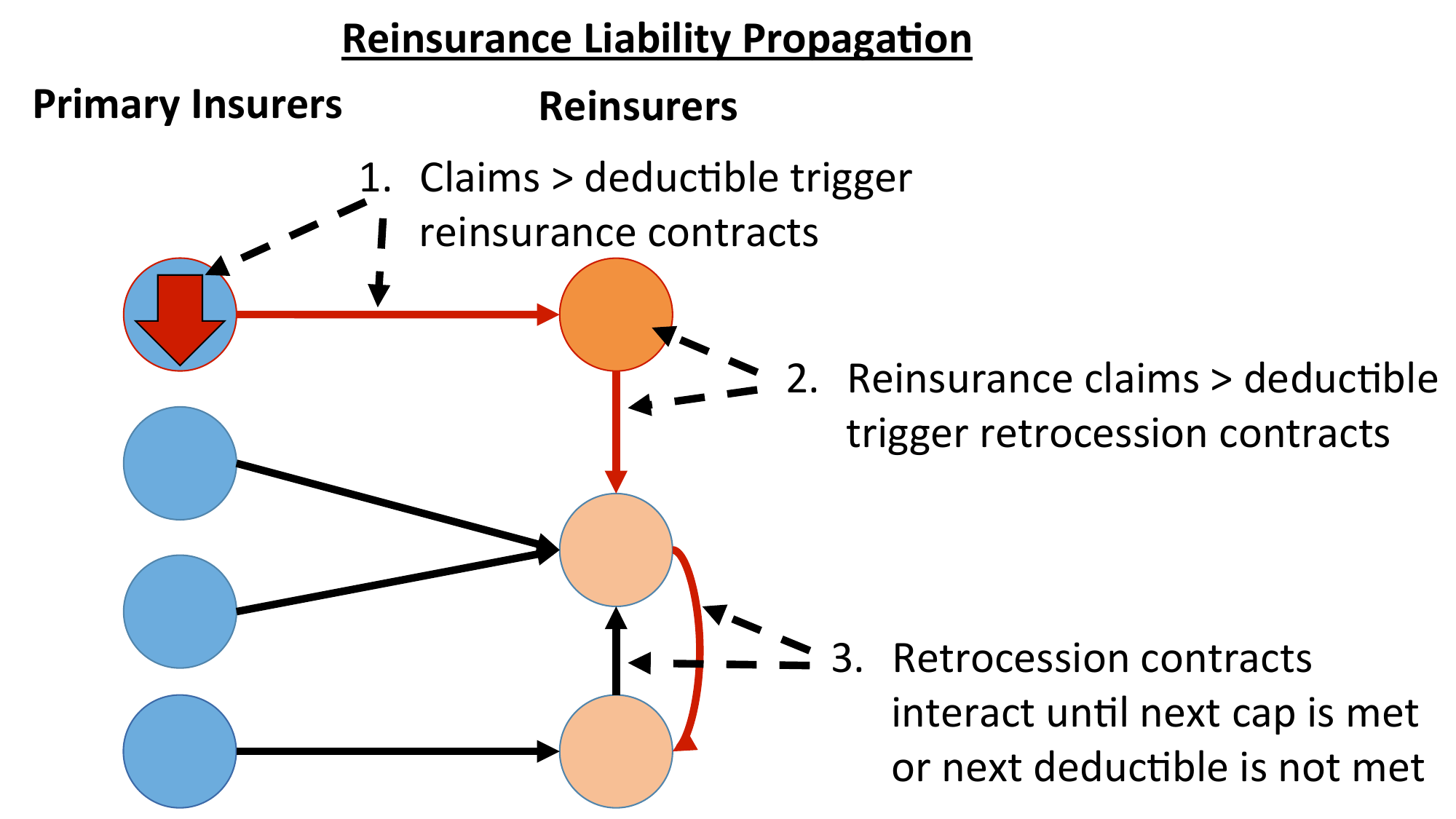}
	\caption{Liability propagation in a reinsurance network.}\label{fig:diagram_liability_propagation}
	\end{subfigure}
	\begin{subfigure}[b]{0.48\textwidth}
	 \includegraphics[width=\textwidth]{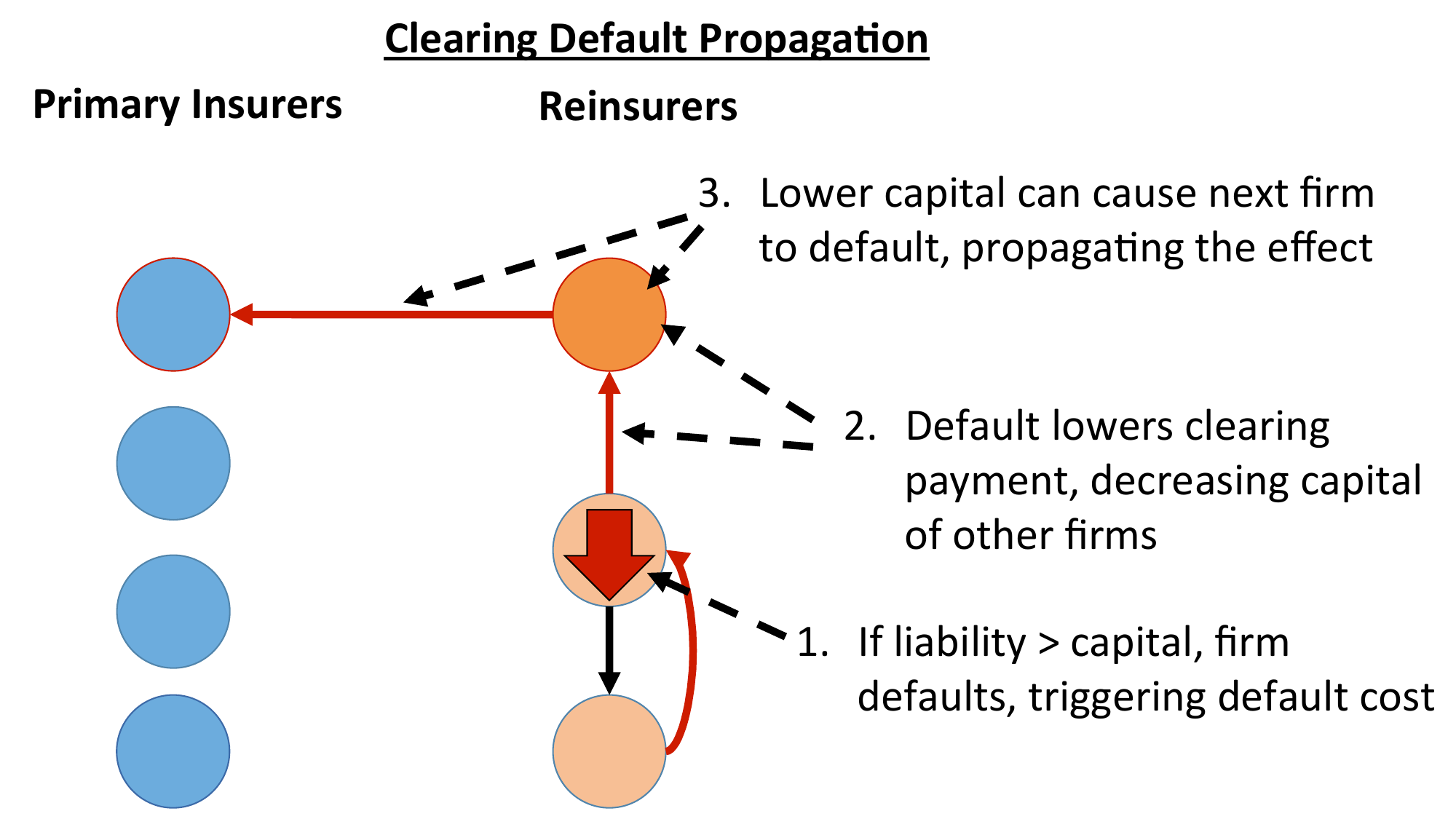}
	\caption{Default propagation in a liability network.}\label{fig:diagram_default_propagation}
	\end{subfigure}
	\caption{Propagation mechanism diagrams.}
\end{figure}

Given a shock, we aim to determine the equilibrium reinsurance payments from a complex interconnection of contracts. Unlike the case of a debt network in \citep{eisenberg01}, we do not outright know the liabilities of each contract, so we cannot directly calculate a clearing payment vector. In a reinsurance network, the liabilities are interrelated and nonlinear. The difficult problem in this case is to determine the equilibrium liabilities given a shock, after which we can solve for a clearing payment vector as in \citep{eisenberg01}. The process for calculating contagion is then as follows:
\begin{enumerate}
\item Given a primary insurance shock, calculate the equilibrium reinsurance liabilities.

\item Apply the available collateral from reinsurance contracts to fulfill or partially fulfill liabilities.

\item Given the remaining capital of firms, clear the remaining liabilities in the network.
\end{enumerate}
We proceed in this paper by developing the machinery to handle the missing piece of the puzzle: the first step. This problem is much more general than related problems formulated in \citep{eisenberg01} and \citep{acemoglu15} and involve matrices with column sums $>1$.

\subsection{Network definitions}

We define the \textbf{reinsurance network} as follows:
\begin{itemize}
	\item $n$ nodes of primary insurance and reinsurance firms
	\item $m$ edges represent reinsurance contracts between firms, directed from reinsurer to reinsured firm. Edges are described by the following weight matrices
	\item $\Gamma$ $n\times n$ matrix of coinsurance rates on contracts (0 if no contract between parties)
	\item $DD$ $n\times n$ matrix of deductibles (also called `attachment points') on reinsurance contracts (0 if no contract between parties)
	\item $CP$ $n\times n$ matrix of reinsurance caps (also called `limits') on contracts (0 if no contract between parties). This is the maximum payout of the contract
	\item $sh$ vector representing shocks to primary insurers.
	\item $e_0$ vector representing initial capital (also called `equity') values of each firm available to payout liabilities
\end{itemize}

We assume the graph is connected, as we can otherwise handle the components separately. We also assume that firms can only reinsure up to $100\%$: i.e., the column sums of $\Gamma$ corresponding to a particular layer of reinsurance sum to $\leq 1$. This is a reasonable assumption as otherwise the contract ceases to serve as insurance and the insured company stands to profit from taking on large losses to their portfolio. This assumption is a standard requirement in insurance contracts.

We will work with the line graph of the network, i.e., the graph that represents edges of the original graph as nodes in the new network and has directed edges when the head of an edge in the original network intersects the tail of another edge in the original network. We define the \textbf{line graph network} as follows:
\begin{itemize}
	\item $m$ nodes representing contracts (i.e., edges) in the reinsurance network
	\item $X$ $m\times m$ adjacency matrix of the line graph, 1-0 weighted
	\item $\ell$ liability vector on contracts
	\item $d$ deductibles vector on contracts
	\item $c$ caps vector on contracts
	\item $s$ shock vector on contracts
	\item $\gamma$ $m\times m$ diagonal matrix of reinsurance rates on contracts
\end{itemize}
The following example describes the transformation to the line graph network.

\begin{example}
	\begin{figure}
		\centering
		\begin{subfigure}[b]{0.35\textwidth}
			\includegraphics[width=\textwidth]{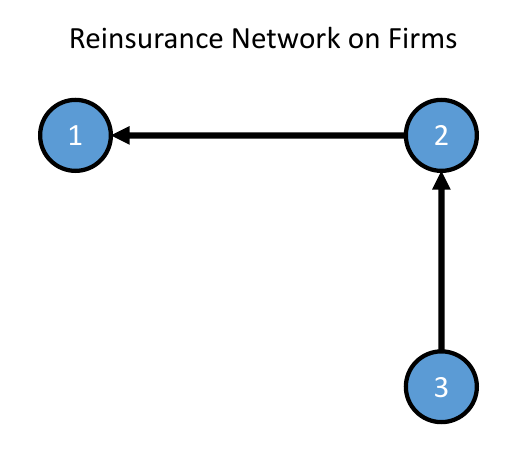}
			\caption{Example reinsurance network.}\label{fig:ex_line_graph_a}
		\end{subfigure}
		\begin{subfigure}[b]{0.35\textwidth}
			\includegraphics[width=\textwidth]{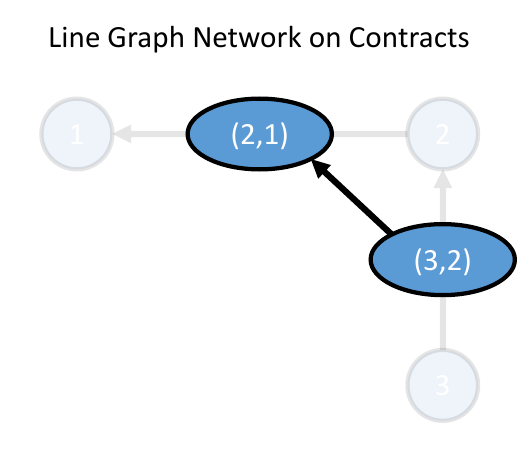}
			\caption{Example line graph network.}\label{fig:ex_line_graph_b}
		\end{subfigure}
		\caption{Example line graph network transformation.}
	\end{figure}
	
	Consider the reinsurance network in Figure~\ref{fig:ex_line_graph_a}. Figure~\ref{fig:ex_line_graph_b} shows the resulting line graph structure. In the original reinsurance network, suppose we have
	$$\Gamma = \begin{bmatrix} 0 & 0 & 0 \\ 0.5 & 0 & 0 \\ 0 & 0.5 & 0 \end{bmatrix}, DD = \begin{bmatrix} 0 & 0 & 0 \\ 10 & 0 & 0 \\ 0 & 10 & 0 \end{bmatrix}, CP = \begin{bmatrix} 0 & 0 & 0 \\ 100 & 0 & 0 \\ 0 & 100 & 0 \end{bmatrix}, sh = \begin{bmatrix} 20 \\ 0 \\ 0 \end{bmatrix}.$$
	Then the line graph network becomes
	$$X = \begin{bmatrix} 0 & 0 \\ 1 & 0 \end{bmatrix}, \gamma = \begin{bmatrix} 0.5 & 0 \\ 0 & 0.5 \end{bmatrix}, d = \begin{bmatrix} 10 \\ 10 \end{bmatrix}, c = \begin{bmatrix} 100 \\ 100 \end{bmatrix}, s = \begin{bmatrix} 20 \\ 0 \end{bmatrix}.$$
\end{example}


The line graph network serves to consider the system as a network of contracts instead of a network of firms. We define a \textbf{financial system} in terms of its line graph network $(X,\gamma, d, c, s)$ (sometimes omitting the $c$ if we are in the domain of infinite caps) as that is the machinery we will need in our theorems and algorithms; however, it can equivalently be defined in terms of the adjacency graphs of the reinsurance network $(\Gamma, DD, CP, sh)$. Note that since $\gamma X$ is nonnegative, the Perron-Frobenius theorem gives us that the spectral radius $\rho(\gamma X) = \lambda_{max}(\gamma X)$. We will show in the next section how to calculate the resulting equilibrium liabilities matrix $L$ (or equivalently liabilities vector $\ell$ in the line graph network) giving liability weights on contracts in a financial system.

\section{Network Liabilities}

\subsection{Liabilities without contract caps}
In the case that each contract has a deductible but no cap (equivalently, each contract has an infinite cap, and so there is no layering of reinsurance), liabilities on contracts equal the sum of direct shocks - deductibles + cross-effects from the network, multiplied by $\gamma$ and with a floor at zero. I.e., the equilibrium liabilities $\ell$ is a fixed point to the equation
$$\Phi(\ell) = \gamma(s + X\ell - d)\vee 0.$$

Define $B(\ell)$ as the $m\times m$ diagonal matrix with 1-0 entries indicating which contracts are activated (i.e., have surpassed the deductible) under $\ell$. Specifically, $B(\ell)_{ii} = 1$ if $(X\ell + s - d)_i \geq 0$ and 0 otherwise. We define a \textbf{$B$-constant set} to be the subset of the domain such that $B$ is a given constant value--i.e., the pre-image of a particular $B$. We will mostly work with $B$-constant sets, so we will refer to $B(\ell)$ as simply $B$. With this terminology, $\Phi$ is equivalent to
$$\Phi(\ell) = \gamma B(X\ell + s - d).$$

Note that this $\Phi$ is nonnegative, monotone increasing (i.e, nondecreasing), and convex as it is the composition of an increasing affine function and a nonnegative, increasing convex function (pointwise maximum). This instance of the problem is similar to the problem considered in \citep{eisenberg01} but without a general upper bound. We solve it in similar ways but provide more direct proofs.

\subsection{Liabilities with contract caps}

In the more general case in which each contract has a deductible and a cap (possibly infinite), there can be multiple layers of reinsurance. Adding the capping effect to the setup started above, the liabilities vector $\ell$ is a fixed point to the equation
$$\Phi(\ell) = \Big(\gamma(X\ell + s - d)\vee 0\Big) \wedge c.$$

We define the following
\begin{itemize}
	\item $C(\ell)$ is the $m\times m$ diagonal matrix with 1-0 entries indicating which edges have surpassed their caps (and so no longer activated); specifically, $C(\ell)_{ii} = 1$ if $\gamma(X\ell + s-d)_{ii} \geq c_{ii}$ and $0$ otherwise.
	\item $\Psi(\ell)$ is a map to a system on the zero diagonal coordinates of $C(\ell)$. Essentially, $\Psi(\ell)$ is $(I-C(\ell))$ where we have dropped the zero rows.
	\item Dropping dependence on $\ell$, $\tilde \gamma = \Psi \gamma \Psi^T$, $\tilde B = \Psi B \Psi^T$, $\tilde X = \Psi X \Psi^T$,
	\item $\bar \ell = C c$, $\tilde v = \Psi(X \bar \ell + s - d)$.
\end{itemize}
We define a \textbf{$(B,C)$-constant set} to be the subset of the domain such that both $B$ and $C$ are given constant values--i.e., the intersection of the pre-image of a particular $B$ and the pre-image of a particular $C$. We will mostly work with $(B,C)$-constant sets, so we will refer to $C(\ell)$ and $\Psi(\ell)$ as simply $C$ and $\Psi$. With this terminology, $\Phi$ is equivalent to
$$\begin{aligned}
\Phi(\ell) &= (I-C)\gamma B \Big( X(I-C)\ell + XCc + (I-C)(s-d)\Big) + Cc \\
&= \Psi^T \tilde \gamma \tilde B (\tilde X \Psi \ell + \tilde v) + \bar \ell.
\end{aligned}$$

Unlike the simpler $\Phi$ without contract caps, which is a subcase  of the more general setting, the $\Phi$ with contract caps is not generally convex. It remains monotone increasing, however. This problem is similar to the problem considered in \citep{acemoglu15}; however, their methodology is limited to the case in which  column sums for the network interaction matrix are $\leq 1$. In the general case of reinsurance layering, column sums of $\gamma X$ can be arbitrarily high. We develop machinery to handle this much more general setting.

Unless specifically pointed out, we will work with the general form of $\Phi$ with contract caps.

\subsection{Unique fixed point}

We characterize conditions under which a unique fixed point exists in Theorem~\ref{prop:unique}. To construct the proof, we will need the following lemmas.

\begin{lemma}\label{lemma:bertsekas}
	A linear system with matrix $A$ is a contraction with respect to some norm if and only if the spectral radius $\rho(A) < 1$. Further, this norm $\|\cdot\|_s$ can be taken to be a weighted Euclidean norm of the form $\|y \|_s = \| M y \|_2$, where $M$ is a square invertible matrix. 
\end{lemma}
A proof of Lemma~\ref{lemma:bertsekas} can be seen in, for example, Appendix B of \citep{bertsekas13}.

\begin{lemma} \label{lemma:BC-constant_sets}
	$(B,C)$-constant sets are convex and form a finite partition of the space $\{ \ell \vert \ell \geq 0\}$.
\end{lemma}

See \hyperlink{pf:BC-constant_sets}{proof in the Appendix}.

We now define the terms used in the theorem:
\begin{itemize}
	\item Let $\mathcal{K}(X,\gamma,d,c,s)$ be the set of $(B,C)$ pairs such that the $(B,C)$-constant set is nonempty. I.e., for $(B,C)\in \mathcal{K}$, there is a feasible $\ell$ such that $B(\ell) = B$ and $C(\ell) = C$. Notice that there is no feasible $\ell$ such that $B(\ell)=0$ and $C(\ell)=I$ as the activation of all caps means that all deductibles are also met. Contracts in different layers reinsuring the same firm also cannot simultaneously be activated: we only reach the second layer if the first layer has reached its cap. Additionally, unless all caps are infinite, $B=I$ and $C=0$ is not feasible.
	
	\item Let $\Omega(X,\gamma,d,c,s)$ be the element-wise maximum over all $(B,C) \in \mathcal{K}$ of the matrices \\$(I~-~C) \gamma B X (I~-~C)$. Notice that $(I-C)$ performs the same function as the $\Psi$ map here; however, it maintains zero rows and columns, making the result comparable across different $(B,C)$ pairs. Notice that
	$$(I-C) \gamma B X (I-C) = \Psi^T \tilde \gamma \tilde B \tilde X \Psi.$$
	We will use this tilde notation to simplify the algebra. To distinguish between tilde notation from different $\Psi(C)$, we will use different subscript notations. E.g., For $(B_1,C_1)\in \mathcal{K}$, $\Psi_1~:=~\Psi(C_1)$ and $\tilde \gamma_1 := \Psi_1 \gamma \Psi_1^T$.
\end{itemize}

\begin{theorem} \label{prop:unique}
	Let $(X,\gamma,d,s,c)$ be a financial system and
	$\Omega := \wedge_{(B,C)\in \mathcal{K}} (I-C)\gamma BX(I-C)$
	be the matrix element-wise maximum. Then if $\rho(\Omega)<1$, there is a unique fixed point to $\Phi(\ell; X,\gamma,d,s,c)$.
\end{theorem}

The idea behind this condition is that not all $(B,C)$ pairs are feasible--in particular, if some caps are finite, we will never have to work with all of $\gamma X$ at once--and so we only need to consider the worst cases of the feasible pairs to construct a dominating linear map. Then if the dominating linear map gives a contractive norm on the whole space, the Banach fixed point theorem gives us uniqueness. The \hyperlink{pf:unique}{proof is provided in the Appendix}.

The following corollary describes a more intuitive condition on the spectral radius of the `full' graph $\gamma X$. However, this simpler condition does not cover general layering structure. In particular, column sums are restricted to $\leq 1$. Notice that the only condition of the corollary is that $\rho(\gamma X)<1$, which further means that any such system leads to a unique fixed point for every possible shock.
\begin{corollary} \label{cor:unique}
	Given a financial system $(X,\gamma, d, c, s)$, if the spectral radius $\rho(\gamma X) < 1$, there is a unique fixed point to $\Phi$.
\end{corollary}

In the case of all infinite caps (i.e., effectively no caps), Corollary~\ref{cor:unique} gives a result similar to Theorems 1 and 2 in \citep{eisenberg01}. We note that our proof is more direct and general. The more general condition on the spectral radius of $\gamma X$ enables the more powerful Banach fixed point theorem to prove the result directly. This method can apply to a broader set of problems, whereas the proof in \citep{eisenberg01} requires minutiae of the specific contagion mechanism and the relations between different firms to arrive at the result.

The general case of Corollary~\ref{cor:unique} is similar to Proposition~1 in \citep{acemoglu15}, which considers clearing vectors in a liabilities network with external senior debt. However, the theorem in \citep{acemoglu15} only applies to matrices of connectivity with all entries strictly positive since the proof relies on the positive version of the Perron-Frobenius theorem. To handle non-negative matrices, we need to require the spectral radius be $<1$ since eigenvalues can otherwise be 1.

Theorem~\ref{prop:unique} and the results we derive in the following sections venture well beyond the setting and results in \citep{acemoglu15} to describe fixed points that apply for the full range of layering structure that can be seen in reinsurance networks. In particular, we need to allow column sums of $\gamma X$ to be $>1$ since there can be multiple complete layers of reinsurance.

A natural question is whether we can strengthen Theorem~\ref{prop:unique} to a wider setting. Conditioning on $\rho(\Psi \gamma B X \Psi^T) < 1$ for all $(B,C)$ pairs that partition the domain into non-empty $(B,C)$-constant sets (and recalling that $C$ defines $\Psi$), $\Phi$ is everywhere a local contraction--i.e., $\Phi$ is a contraction restricted to each $(B,C)$-constant set by some metric. We can further show that $\Phi$ is globally nonexpansive. We conjecture that, under these conditions, $\Phi$ has a unique fixed point. However, this problem is more challenging because we need to establish a metric over which the function is globally contractive in order to use the existing machinery. We leave this as further work.

\subsection{Other cases: unique, multiple, and no fixed points}

Problematic graph structure can cause $\Phi$ to be non-contractive. This occurs when circular sequences of contracts allow 100\% reinsurance to be continually recirculated through a given set of nodes. We will refer to an instance of this as a `\textbf{100\% cycle}'. Cycle here refers to the graph theoretic meaning as opposed to the economic meaning. Figure~\ref{fig:cycles} provides three examples of how this can happen. Figure~\ref{fig:cycle1} is the simplest example that directly recirculates 100\% reinsurance around one cycle. Figure~\ref{fig:cycle2} shows that multiple cycles can interact to recirculate 100\% reinsurance to a central node. Figure~\ref{fig:cycle_complete} shows that in the most extreme case of a complete graph with all $\Gamma=1/(n-1)$, 100\% reinsurance can be recirculated to every node in the network: as all weights are $1/2$, the reinsurance that can be recirculated to each node can be a geometric sequence that converges to 1.

\begin{figure}
	\centering
	\begin{subfigure}[b]{0.4\textwidth}
		\includegraphics[width=\textwidth]{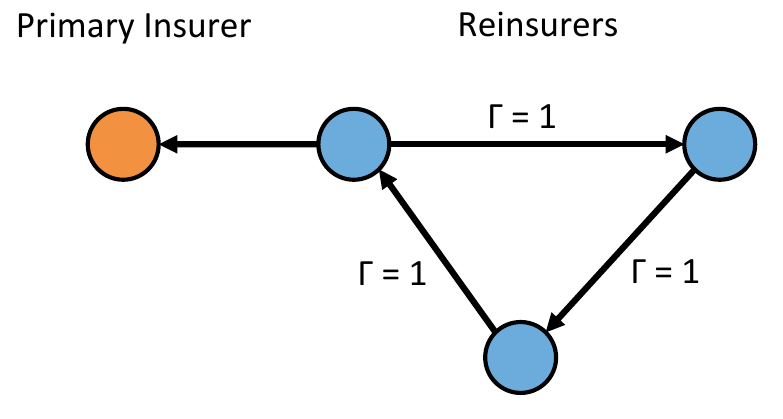}
		\caption{A direct 100\% cycle.}\label{fig:cycle1}
	\end{subfigure}
	~ 
	\begin{subfigure}[b]{0.4\textwidth}
		\includegraphics[width=\textwidth]{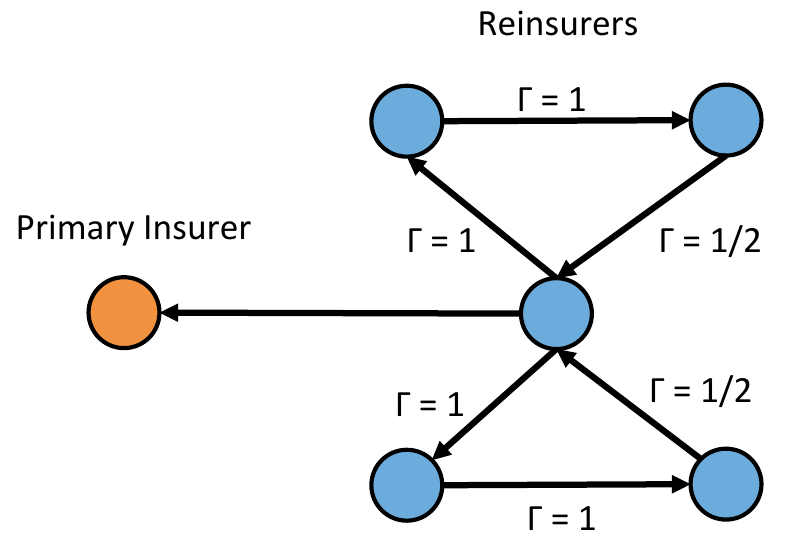}
		\caption{100\% cycle from two interacting cycles.}\label{fig:cycle2}
	\end{subfigure}
	~ 
	\begin{subfigure}[b]{0.5\textwidth}
		\includegraphics[width=\textwidth]{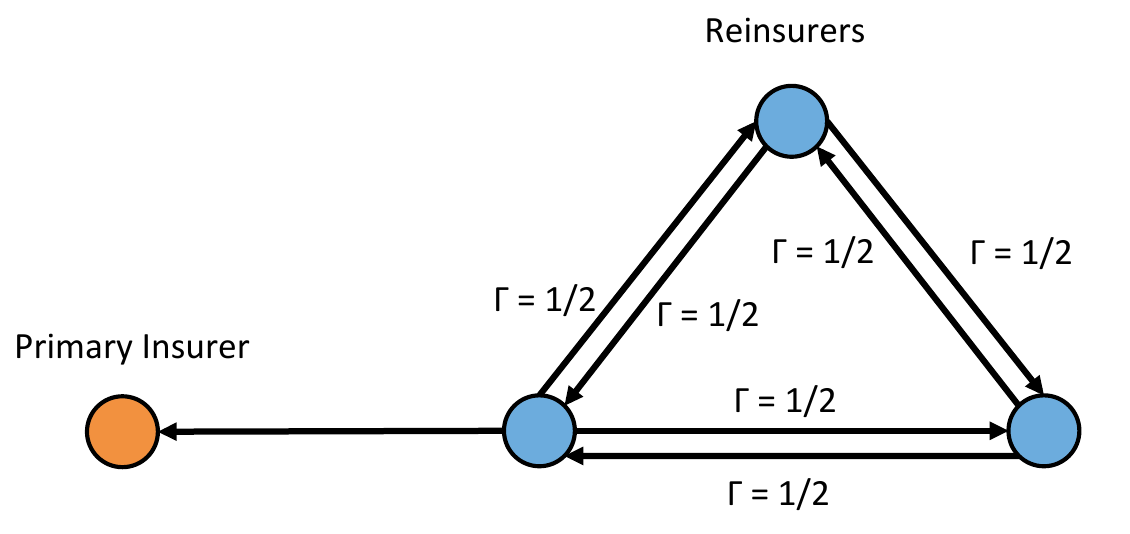}
		\caption{Many interacting cycles can form a 100\% cycle.}\label{fig:cycle_complete}
	\end{subfigure}
	\caption{Some examples of 100\% cycles.}\label{fig:cycles}
\end{figure}

\subsubsection{100\% cycles: no caps case}
For simplicity, we first describe the effects of these 100\% cycles from the perspective of a system with infinite/no caps, in which $\gamma X$ mostly describes the entire system.

Analytically, these 100\% cycles cause the matrix powers $(\gamma X)^k$ to fail to converge to 0 as $k \rightarrow \infty$ since we enter an infinite increasing loop. On the other hand, the condition on the spectral radius $<1$ from Corollary~\ref{cor:unique} ensures that $\lim_{k\rightarrow \infty} (\gamma X)^k = 0$. Checking the spectral radius is a simple check of whether a 100\% cycle exists; however, it may be difficult to identify the actual cycle in the network. As shown by the examples in Figure~\ref{fig:cycles}, a problematic cycle can be a complex interaction of many graph cycles. A naive method to search for the problematic cycle would involve iterating over graph cycles in the network, which is itself NP-hard.

If $\rho(\gamma X) \geq 1$ in this setting, a 100\% cycle exists. In this case, there may still be a unique fixed point. If not, there may be a smallest fixed point or there may be no fixed point. The following characterize these cases and follow from the main theorems in this section.

\paragraph{Unique fixed point if 100\% cycles are not activated.} The term $B$ in $\gamma BX$ serves to remove edges from the resulting graph. Since the spectral radius of a proper subgraph is less than the spectral radius of the initial connected graph, we may have local contraction on some $B$-constant sets, but not on the whole domain. Depending how far the shock $s$ spreads, there may still be a unique fixed point on the contractive $B$-constant sets. In this case, 100\% cycles that cause the spectral radius to be $\geq 1$ are not activated. We can restrict the domain to the contractive region to find a unique fixed point. Notably, such a system will not yield a fixed point for all possible shocks.
	
\paragraph{No fixed point if 100\% cycles are activated.} When 100\% cycles are activated by shocks, the system is non-contractive. In this case, there is no fixed point, and some nominal liabilities will increase to infinity as contracts call circularly in 100\% cycles. For example, say that the primary insurer faces a shock of 10 in Figure~\ref{fig:cycle1}. Then this loss is passed around the cycle because the reinsurance rate is 100\%. The first reinsurer faces a loss of 20, which is again passed around the cycle, and so on. Since $\Phi$ is monotone, these ``infinite'' fixed points are the only cases of nonexistence. This could happen but is extremely unlikely in practice. The only contracts without caps are proportional and not reinsured at 100\%. A 100\% cycle in this case would be quite contrived.
	
\paragraph{Multiple fixed points from self-reinforcing claims.} If 100\% cycles have their deductibles exactly met from outside claims but are otherwise unactivated, we have multiple solutions. A simple example of this is a 100\% cycle with zero deductibles and zero outside claims. Any value on these contracts is self-fulfilling, and so there are infinitely many fixed points. Notice that each 100\% cycle is self-contained due to the assumption that no firm can reinsure over 100\%. This means that the self-fulfilling solutions in each such cycle are independent. Then the set of fixed points looks like a Cartesian product of individual solution subsets related to the 100\% cycles in the network that have deductibles exactly met by outside claims.

\subsubsection{100\% cycles: general caps case}

Corollary~\ref{cor:unique} tells us cases in which we can prove unique fixed points for all shocks. Theorem~\ref{prop:unique} covers more cases. Outside of these theorems, there may still be unique fixed points in other systems: in some cases if the shocks do not activate problematic network structure, and in other cases there may be no problematic network structure to worry about. When we are not guaranteed uniqueness, we can see similar effects as in the no caps case of multiple fixed points or no fixed points. The following characterize these situations and follow from the main theorems in this section.

\paragraph{Caps limit non-contractive effects.}
Even if 100\% cycles are activated, finite contract caps can limit the non-contractive effect to particular layers that reach capacity. In this case, the finite caps remove the problematic graph structure from the remainder of the problem.

\paragraph{Multiple fixed points bounded within layers.}
If a 100\% cycle has deductibles exactly met from outside claims but remains otherwise unactivated, we have multiple solutions as before. However, these solutions are constrained by the caps and restricted to self-reinforcing liabilities within particular layers of reinsurance.

As we will see in the following subsections, in the case of multiple fixed points, there is always a least fixed point that represents the real world solution. Other fixed points are mathematical artifacts from self-reinforcing liabilities that are not propagations of primary insurance claims.

\paragraph{No fixed point if uncapped 100\% cycle is activated.}
As before, the only instances of nonexistence are when fixed point iteration diverges to infinity on some liabilities. This only happens when an uncapped 100\% cycle is activated. Since $\Phi$ is monotone, all other problematic structures are eventually capped out and so do not contribute to nonexistence. For the same reasons as before, this is extremely unlikely to occur in practice.

\subsection{Least fixed points}

In the event that we have multiple solutions, there is a least fixed point as the others are self-fulfilling and not caused by actual claims. This is formalized in the following theorem.

\begin{theorem} \label{prop:kleene}
	For financial system $(X,\gamma, d, c, s)$, if a fixed point of $\Phi(\ell; X,\gamma,d,c,s)$ exists, then there is a least fixed point. Further, fixed point iteration starting at 0 converges to the least fixed point.
\end{theorem}

Note that the Tarski fixed point theorem cannot be used in this setting because the domain lattice is not necessarily bounded above. Instead, our proof relies on the Kleene fixed point theorem. The Kleene fixed point theorem is additionally constructive, and so guarantees that fixed point iteration converges to the least fixed point if it exists. We present the \hyperlink{pf:kleene}{proof of Theorem~\ref{prop:kleene} in the Appendix}, including an overview of the Kleene fixed point theorem.

The next theorem gives us a conditions under which we can apply the Tarski fixed point theorem, in which case we know outright that a least (and greatest) fixed point exists.

\begin{theorem}\label{prop:tarski}
	Let $(X,\gamma, d, c, s)$ be a financial system. Let $\Psi_0$ map to a system on the edges with infinite caps (or map to the zero matrix if all caps are finite). Then if the spectral radius $\rho(\Psi_0 \gamma X \Psi_0^T) < 1$, $\Phi$ has least and greatest fixed points.
\end{theorem}

See \hyperlink{pf:tarski}{proof in the Appendix}.

Note that this means that, if all caps are finite, there is always a maximum fixed point. This means that liabilities cannot spiral to infinity. In practice, most reinsurance contracts have caps, and so there will be a resulting equilibrium liability structure.

We have provided a wide variety of cases for which we are guaranteed least fixed points. In the next subsection, we show that the least fixed point is the real world solution of interest.

\subsection{Multiple fixed points: net liabilities equal}

Recall that $L$ is the equilibrium firm-to-firm liabilities matrix. Rows represent liabilities from the row firm to each column firm. Define the net liabilities of each firm as a vector
$$\Delta(L) := L^T e - Le,$$
where $e$ is the all ones vector. $(L^T e)_i$ is what $i$ is due from other firms. $(L e)_i$ is what $i$ owes to other firms. The following theorem shows that net liabilities of firms are constant across multiple fixed points of $\Phi$.

\begin{theorem} \label{prop:net_liabilities_equal}
	If $L,L'$ are fixed points of $\Phi$, then $\Delta(L) = \Delta(L')$. I.e., the net liabilities of each firm are equivalent under any fixed point.
\end{theorem}

We will need the following two lemmas to prove the theorem.

\begin{lemma} \label{lemma:net_liabilities_1}
	If $L,L'$ are fixed points of $\Phi$ with $L \geq L'$ (entry-wise), then $\Delta(L) \leq \Delta(L')$.
\end{lemma}

\begin{lemma} \label{lemma:net_liabilities_2}
	If $L$ is a fixed point of $\Phi$, then $\sum_i \Delta_i(L) = 0$.
\end{lemma}

See \hyperlink{pf:net_liabilities_1}{proof of Lemma~\ref{lemma:net_liabilities_1}} and \hyperlink{pf:net_liabilities_2}{proof of Lemma~\ref{lemma:net_liabilities_2}} in the Appendix. Lemma~\ref{lemma:net_liabilities_2} is rather immediate because the reinsurance system does not amplify losses--it only distributes losses across the network. As the initial shock is not included in $L$, the terms in $L$ sum to 0. Once we have established these lemmas, the \hyperlink{pf:net_liabilities_equal}{proof of Theorem~\ref{prop:net_liabilities_equal}, which is included in the Appendix}, is similar to that of Theorem~1 in \citep{eisenberg01}.

As net liabilities are equivalent between fixed points, the least fixed point corresponds to the real world solution. It represents the propagation of primary insurance shocks, whereas other fixed points add self-fulfilling liabilities on top of this. To see this, note that fixed point iteration starting from zero represents stepwise propagation of primary insurance shocks and converges to the least fixed point. Then, from Theorem~\ref{prop:net_liabilities_equal}, any greater fixed points can only come from adding additional liabilities that net out.

\subsection{Consequences of multiple fixed points}

Note, however, that Theorem~\ref{prop:net_liabilities_equal} does not imply that the \emph{clearing} of different fixed point liabilities are equivalent. In general, the clearing will depend on the \emph{nominal} liabilities as opposed to the net liabilities. Figure~\ref{fig:clearing_diff} gives an example where different fixed points lead to different clearing outcomes. The liability on edge $(B,A)$ is $L_{BA} = 10$ in any fixed point. However, $L_{CB}=L_{BC}=0$ and $L_{CB}=L_{BC}=10$ are both valid fixed points (the net liabilities are the same for both fixed points
$\Delta L_A = -10$, $\Delta L_B = 10$, $\Delta L_C = 0$). If the capital of firm $B$ has zero value, $A$ will receive zero payment after clearing in the minimum fixed point whereas in the $L_{CB}=L_{BC}=10$ fixed point, $A$ will receive $5$ and $C$ will pay $5$.

\begin{figure}
	\centering \includegraphics[width=8cm]{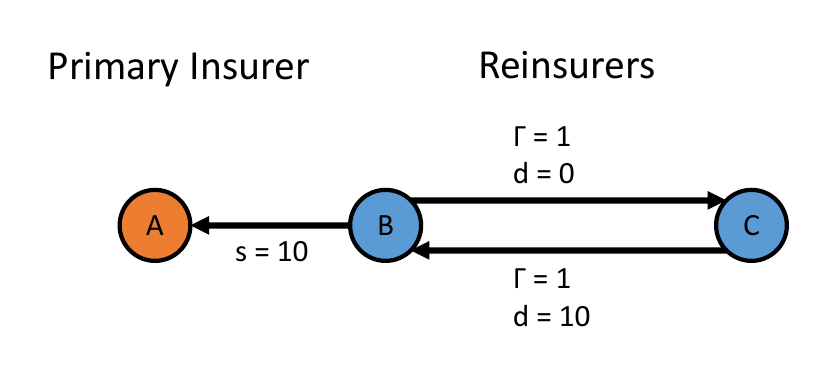}
	\caption{Example in which different fixed points lead to different clearing payments.}\label{fig:clearing_diff}
\end{figure}

The only reason a system would find itself in a non-least fixed point is from fraud. In this case, someone has inflated their liabilities such that the network makes it self-fulfilling. In the case of Figure~\ref{fig:clearing_diff}, firm $A$ is better off in a non-minimal fixed point and could have an incentive to influence (outside of our model) firm $B$. Firm $B$ could present fraudulent higher claims to firm $C$, which would be self-fulfilling from the 100\% cycle. This would cause a non-minimal fixed point in liabilities. If the network is complicated, this fraud could be difficult to uncover. 

\begin{remark}
There are two other realistic mechanisms that can cause multiplicity of solutions, which pose governance challenges:
\begin{itemize}
	\item The parameters of many reinsurance contracts are not well defined, even to the parties of the contracts. It is a common practice in the reinsurance industry to agree to `in the future agree on a specific contract'.\footnote{Private conversations with an insurance industry executive. All errors are our own.} In extreme cases, these `contracts' have been litigated to determine what contract would have been reasonably agreed upon. In this case, the global contract parameters are not in principle knowable, as assumed in our model, and there are additional potential solutions for the different potential versions of the unknown contract.
	
	\item Given liabilities, multiple fixed points for determining clearing payments can also exist when costs of default are nonzero \citep{veraart13}. This is realistically the case as there are legal, transactional, and liquidity costs associated with real world defaults.
\end{itemize}	

In either case, if there is disagreement in the payouts of reinsurance contracts, such as from multiple potential solutions, the issue goes to a panel of arbitors to resolve \citep{reins_arbitration}. The members of the panel are typically active or former executive officers of insurance or reinsurance companies \citep{reins_arbitration} and will have different incentives. For example, these could include the following: driving a competitor out of the market, limiting contagion to given markets, or pinning default on parties that are least connected to themselves. Even when the arbitors do not have direct conflicts of interest, indirect conflicts of interest are unavoidable through network structure. The arbitors will have different perceived risk exposure to the various solutions. These incentives are outside of the focus of this study; our purpose is to illustrate that cases like this can happen. We leave it to future work to model these incentives and design good governance structures to account for these.
\end{remark}

We have widely characterized least fixed points of $\Phi$. In the next section, we provide efficient algorithms for finding these fixed points.

\subsection{Algorithms to find the least fixed point}

To find the minimum fixed point, if it exists, we can perform a fixed point iteration of $\Phi$ starting at 0. Algorithm~\ref{alg:push} performs this fixed point iteration. The constructive statement of the Kleene fixed point theorem guarantees that Algorithm~\ref{alg:push} converges to the minimum fixed point, if it exists. In practice, this runs efficiently. However, in the worst case, it can take arbitrarily long. To see this, consider an arbitrarily small damping of the near-100\% cycle from Figure~\ref{fig:damped_cycle} with infinite caps. The fixed point iteration operates by pushing mass iteratively around the system starting at the primary liabilities. As a given edge's liabilities increase in one iteration, its tail node makes further calls on its reinsurers in the next iteration, increasing the liabilities on their edges. The number of fixed point iterations goes to infinity as the damping goes to 0. This is because all excess loss will end up with the damping node, but we require arbitrarily many trips through the cycle to reach equilibrium since the mass removed from the recirculation in each iteration is smaller with a smaller damping.

\begin{algorithm}[H]
	\algorithmicrequire $d$, $c$, $\gamma$, $s$, $X$
	\begin{algorithmic}
		\State Let $\ell_0$ be the zero vector, $t\leftarrow 1$, $\tt{finish}\leftarrow \tt{False}$
		\While {$\tt{finish} = \tt{False}$}
		\State Let $b_t$ indicate the entries that satisfy $X \ell_{t-1} + s - d \geq 0$; define $B_t = \text{diag}(b_t)$
		\State $\ell_t \leftarrow \text{min} \Big(B_t \gamma (X \ell_{t-1} + s - d), c \Big)$ element-wise
		\If {$\ell_{t} = \ell_{t-1}$}
		\State $\tt{finish} \leftarrow \tt{True}$
		\EndIf
		\State $t \leftarrow t + 1$
		\EndWhile
		\State \Return $\ell_t$
	\end{algorithmic}
	\caption{Fixed point iteration algorithm to determine reinsurance network liabilities} \label{alg:push}
\end{algorithm}

\subsubsection{Iterative linear solver: no caps case}

This motivates Algorithm~\ref{alg:deductibles_iter}, which calculates liabilities in polynomial time complexity by iteratively solving linear systems in networks without caps. The linear systems are of the form $\ell = \gamma B(s + X\ell -d)$ which has a unique solution
$$\ell = (I-\gamma BX)^{-1} \gamma B(s-d)$$
if $(I-\gamma BX)$ is nonsingular. In the case of zero deductibles (i.e., proportional contracts), costs are shared linearly according to coinsurance rates between firms, and only one iteration is required.

\begin{algorithm}[H]
	\algorithmicrequire $d$, $\gamma$, $s$, $X$
	\begin{algorithmic}
		\State Let $B_0$ be the 0 matrix and $t=1$
		\State Define diagonal matrix $B_1$ by setting $B_{1,ii} = 1$ if $(s-d)_i \geq 0$ and 0 otherwise
		\While {$B_t \neq B_{t-1}$}
		\State Solve for $\ell_t = \gamma B_t (s + X\ell_t - d)$
		\State $t \leftarrow t + 1$
		\State Define diagonal matrix $B_t$ by setting $B_{t,ii} = 1$ if $(s + X\ell_t - d)_i \geq 0$ and 0 otherwise
		\EndWhile
		\State \Return $\ell_t$
	\end{algorithmic}
	\caption{Determine reinsurance network liabilities in a system without caps} \label{alg:deductibles_iter}
\end{algorithm}

\begin{proposition} \label{prop:no_caps_iter_converge}
	Given a financial system $(X, \gamma, d, s)$ with infinite caps, if the spectral radius $\rho(\gamma X) < 1$, then Algorithm~\ref{alg:deductibles_iter} converges to the solution in at most $O(m^4)$ time.
\end{proposition}

See \hyperlink{pf:no_caps_iter_converge}{proof in the Appendix}.

Note that Algorithm~\ref{alg:deductibles_iter} works in some additional cases (still assuming infinite caps). In the first case, we may have problematic graph structure as demonstrated in the previous section, but which is not in a region of the graph that is activated by the given shock. This means that the algorithm never leaves the contractive region of $\Phi$. In the second case, $(I-\gamma BX)$ can be invertible even if $\rho(\gamma BX) \geq 1$. Note that in this latter case, $(I-\gamma BX)^{-1}$ will not be nonnegative, but this is not an issue as $\ell = (I-\gamma BX)^{-1}\gamma B(s-d)$ will still be nonnegative as required.

\subsubsection{Iterative linear solver: caps case}

We now adapt the iterative linear solver to the setting with contract caps. However, this is complicated by the fact that an iteration from 0 could mistakenly activate edges due to overcapacity leakage along edges in one of the linear solves. To avoid this, we need to start from the worst case and iterate downward, which results in a process that terminates at the maximum fixed point. Thus, this process only converges to the least fixed point if there is a unique fixed point.

The linear systems that come up in iterations are of the form $\tilde \ell = \tilde \gamma \tilde B(\tilde X \tilde \ell + \tilde v)$ (recall that tilde notation incorporates the $\Psi$ transformation onto the subsystem of edges that are not overcapacity in the previous iteration), which has a unique solution
$$\tilde \ell = (I-\tilde \gamma \tilde B \tilde X)^{-1} \tilde \gamma \tilde B \tilde v$$
if $(I-\tilde \gamma \tilde B \tilde X)$ is nonsingular.

Algorithm~\ref{alg:deductibles_caps_iter} describes the iterative linear solver.

\begin{algorithm}[H]
	\algorithmicrequire $d$, $c$, $\gamma$, $s$, $X$
	\begin{algorithmic}
		\State Let $b_0$ and $c_0$ be all twos vectors ($>$ all ones vectors) and $t=1$
		\State Let $b_1$ be the all ones vector and $c_1$ indicate entries with finite caps
		\While {$b_t \neq b_{t-1}$ and $c_t \neq c_{t-1}$}
		\State Let $\bar \ell = \text{diag}(c_t) c$
		\State Let $\Psi$ map to a system on the zero coordinates of $c_t$
		\State For $i$ a $c_t$ zero coordinate index, let $\psi(i)$ give the corresponding coordinate index under $\Psi$.
		\State Let $\tilde \gamma = \Psi \gamma \Psi^T, \tilde B_t = \Psi B_t \Psi^T, \tilde X = \Psi X \Psi^T, \tilde v = \Psi(s+X\bar \ell_1 -d)$
		\State Solve for $\tilde \ell = \tilde \gamma \tilde B_t(\tilde X \tilde \ell + \tilde v)$
		\State Let $t \leftarrow t + 1$
		\State $\ell \leftarrow \Psi^T \tilde \ell + \bar \ell$
		\State Let $b_t$ indicate the entries that satisfy $X \ell + s - d \geq 0$
		\State Let $c_t$ indicate the entries that satisfy $\gamma(X \ell + s - d) \geq c$
		\EndWhile
		\State \Return $\ell$
	\end{algorithmic}
	\caption{Determine reinsurance network liabilities in a system with deductibles and caps} \label{alg:deductibles_caps_iter}
\end{algorithm}

To prove that Algorithm~\ref{alg:deductibles_caps_iter} converges, we need a stronger condition than in Theorem~\ref{prop:unique}. This is because we have to start at an upper bound that is easily computable. In some cases, there may be a suitable upper bound, but unlike in Algorithm~\ref{alg:deductibles_caps_iter}, it is not immediately clear what it is. The following proposition provides sufficient conditions for convergence.


\begin{proposition}\label{prop:caps_iter_converge}
	Let $(X,\gamma,d,s,c)$ be a financial system and
	$\Omega := \wedge_{(B,C)\in \mathcal{K}} (I-C)\gamma BX(I-C)$ be the matrix element-wise maximum, and $\Psi_0$ be the map to a system on the edges with infinite caps (or map to the zero matrix if all caps are finite). Then if $\rho(\Omega)<1$ and $\rho(\Psi_0 \gamma X \Psi_0^T)<1$, Algorithm~\ref{alg:deductibles_caps_iter} converges to the solution in at most $O(m^4)$ time.
\end{proposition}

See \hyperlink{pf:caps_iter_converge}{proof in the Appendix}.

Note that the conditions of the proposition may be difficult to check in general. An easier condition to check is that $\rho(\gamma X)<1$; however, this is again not general enough to include many real world cases of multiple layers of XL contracts.

This algorithm additionally `works' when $\rho(\tilde \gamma \tilde B \tilde X) < 1$ for all $(B,C)$ pairs. However, we only know that it terminates at the minimum fixed point (i.e., a unique fixed point) when $\rho(\Omega)<1$.

\section{Real World Implications of the Network Model}

In the previous sections, we developed the machinery of our reinsurance network model. This model works sequentially by calculating liabilities given a shock, applying the collateralized portion of reinsurance contracts to cover liabilities, and then calculating clearing payments in the network. We now discuss two features that result from this model: dangerous network structures and parameter sensitivity. These features introduce new issues in risk management, present a new incentive to combat fraud, and demonstrate the importance of global contract design to ensure the insurance system works well.

\subsection{Dangerous network structures cause reinsurance spirals}

Relaxations of 100\% cycles cause counterintuitive nonlinear behavior known as reinsurance spirals. By relaxation, we mean that the cycle circulates close to (but $\leq$) 100\% reinsurance. We will refer to this as `\textbf{relaxed cycles}'. This is introduced, for instance, in \citep{bain99} with an example in the Lloyd's and London reinsurance markets in the 1980s. In these spirals, nominal liabilities increase at each step through the graph cycle until one of the contract edges reaches its cap, after which all excess liability is left with the reinsured party of that contract. Figure~\ref{fig:spiral} provides an example. In this example, even though the size of the shock is less than all contract caps, the spiral effect causes the cap on $(C,A)$ to be reached, leaving all liability for the shock on $A$. Given local first-degree information, all parties think they are adequately reinsured; however, it turns out that this is not the case. Relaxation of the 100\% cycle to smaller values of $\gamma$ lessen the growth of liabilities in the cycle, but lead to a similar effect, where a disproportionate amount of excess liability is left with a single party. Further, the effect is the same, even if we arbitrarily scale the caps around the cycle. Even if caps are very high, in which case a firm would intuitively expect to be very well reinsured, firms are still subject to the exact same spiral risk.

\begin{figure}
\centering \includegraphics[width=6cm]{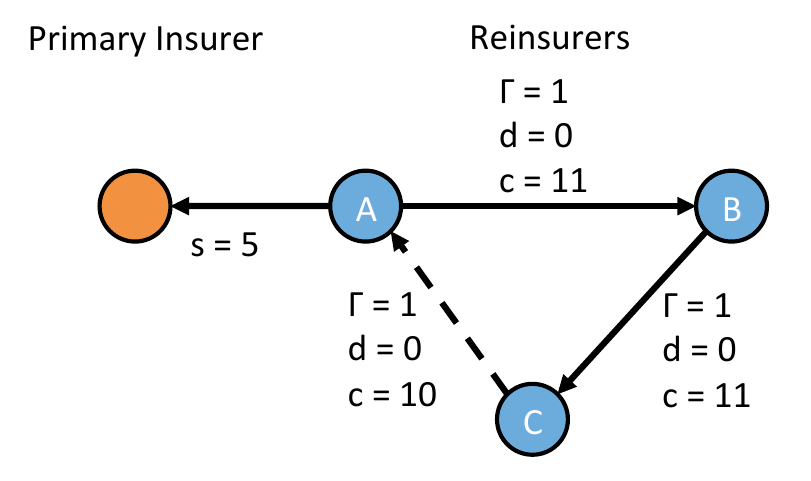}
\caption{Example of a reinsurance spiral. An initial shock of 5 is passed around the reinsurance cycle with 0 deductibles until nominal liabilities reach 10, at which point the cap on edge $(C,A)$ is activated. All excess loss is then left with firm A.}\label{fig:spiral}
\end{figure}

Another type of spiral that can happen is when a relaxed cycle is damped by a node outside the cycle. See Figure~\ref{fig:damped_cycle} for an example. In this case, in each trip through the cycle, a small proportion of excess liability is siphoned off by the damping reinsurer. The remaining proportion continues around the cycle until a proportion of it is again siphoned off by the damper. In equilibrium, a disproportionate amount of the excess loss can be left on the damping reinsurer, depending on the cap parameters in the cycle. The damping reinsurer may not be aware of the role they are taking in the network; given local first-degree information, they may think they are only reinsuring one firm instead of the whole cycle.

If there are multiple damping reinsurers, it will be difficult to predict whether one of the dampers (and which one) will be left with disproportionate excess liability if there is imperfect information about network parameters. For instance, one damping contract could have a low cap, leaving most liability on the second damping contract. Alternatively, a contract cap within the cycle could be activated and leave most excess loss on one of the reinsurers in the cycle. We can also have a damping chain of reinsurers. In this case, a node can be arbitrarily far from a relaxed cycle in a connected graph but still be left with disproportionate excess loss.

\begin{figure}
\centering \includegraphics[width=9cm]{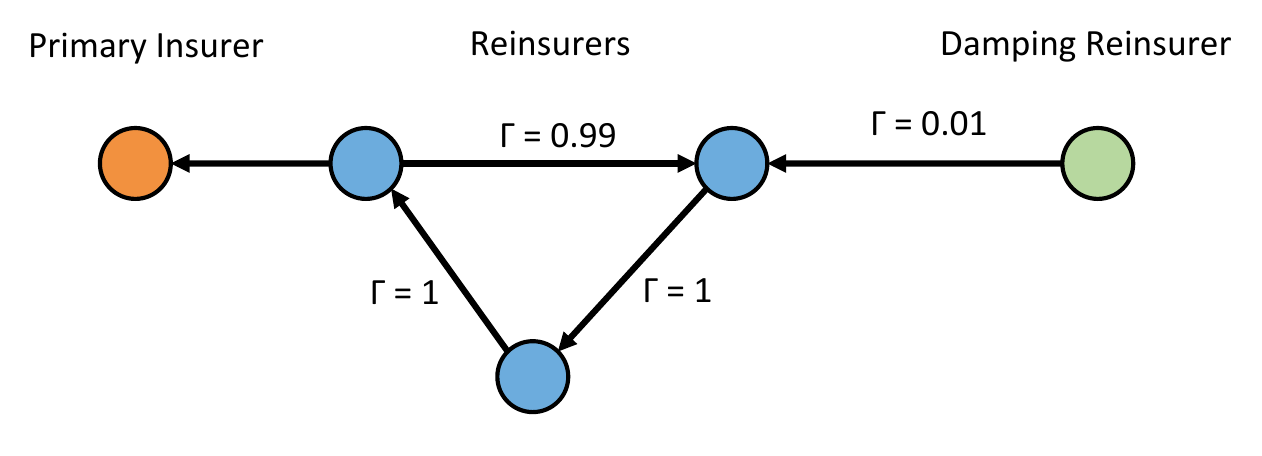}
\caption{Example of a relaxed cycle with a damping reinsurer. Excess loss is passed around the cycle with 1\% being absorbed by the damping reinsurer in each circulation. If the caps on the cycle contracts are high, disproportionate excess loss is left on the damping reinsurer.}\label{fig:damped_cycle}
\end{figure}

\textbf{Standard risk management does not work for relaxed cycles} Relaxed cycles can serve to aggregate losses from multiple sources across the network in a way that is not transparent to a damping reinsurer who only knows its local structure in the network. Figure~\ref{fig:damped_cycle3} is an example where losses from multiple primary insurers are aggregated through the relaxed cycle, leaving all excess liability with the damping reinsurer. For comparison, a tree structure such as in Figure~\ref{fig:tree} can aggregate losses from many primary insurers onto one reinsurer (the root node). However, reinsurers can control for this tree aggregation risk by putting limits on the size of the reinsurered portfolios (usually in terms of premiums the reinsured firm receives). The reinsured portfolios would have to be large to lead to large aggregations of losses. This risk management method does not work in the case of relaxed cycles. Even if firms in a relaxed cycle individually have small reinsurance portfolios, the relaxed cycle can include a large number of firms and the spiraling behavior can aggregate losses from all of these portfolios.

\begin{figure}
	\centering
	\begin{subfigure}[b]{0.58\textwidth}
		\includegraphics[width=\textwidth]{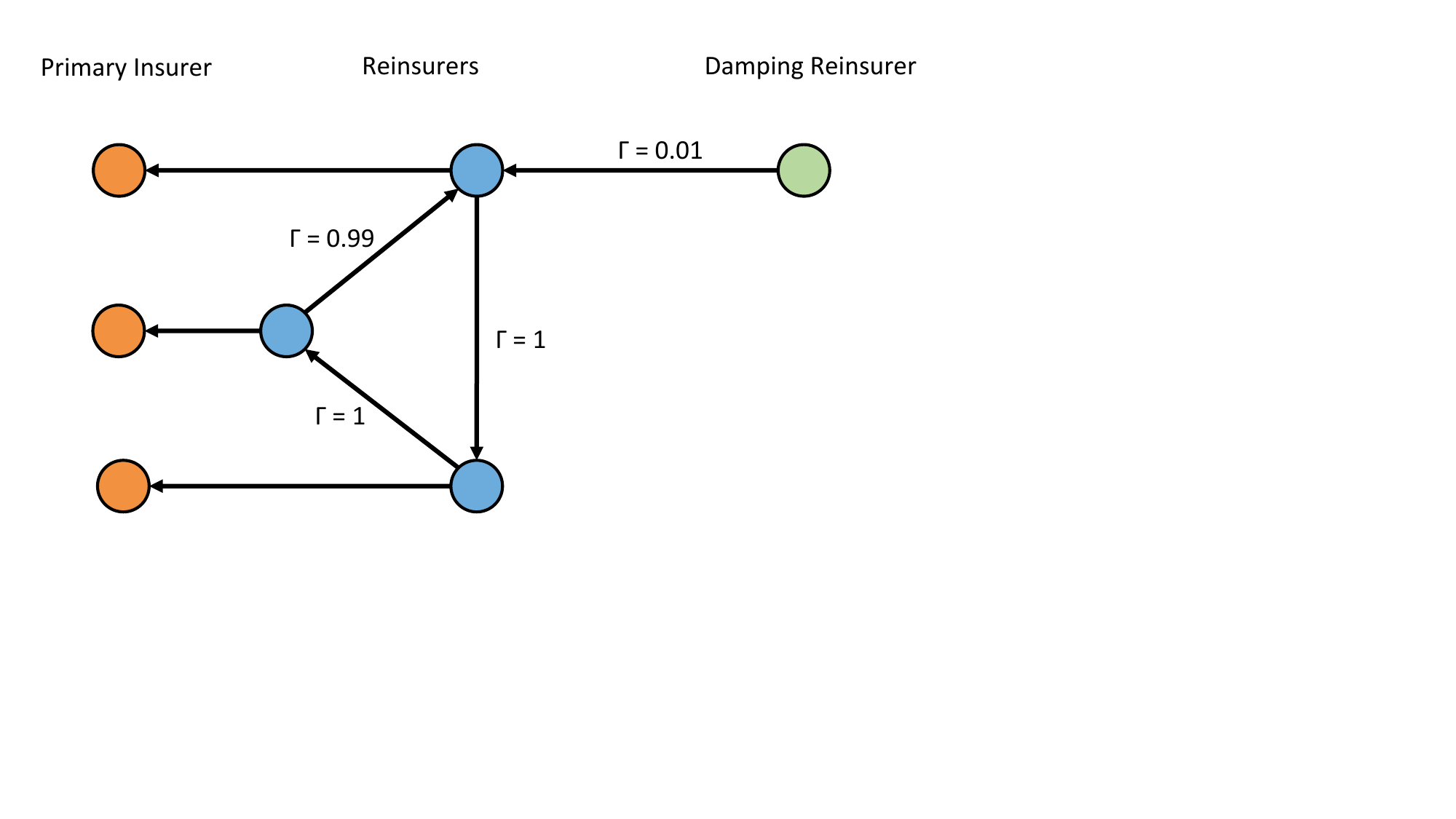}
		\caption{Example of cycle that aggregates multiple losses. With \\ high caps, most aggregated loss is absorbed by damper.}\label{fig:damped_cycle3}
	\end{subfigure}
	\begin{subfigure}[b]{0.4\textwidth}
		\includegraphics[width=\textwidth]{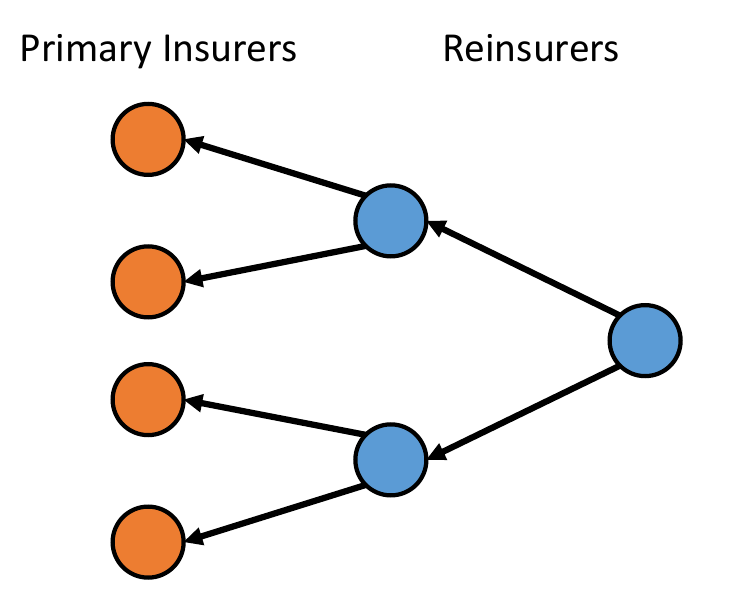}
		\caption{A tree structure that aggregates many losses on one reinsurer.}\label{fig:tree}
	\end{subfigure}
	\caption{Standard risk management works for trees but fails in the case of relaxed cycles.}
\end{figure}



Anther example of unintuitive behavior is that a contract cap does not necessarily limit the liability of a node caused by another node. Figure~\ref{fig:cap_negated} gives an example where graph connections through a second layer of reinsurance counteract the contract cap on the first layer. In this case, $B$ reinsures $A$ up to the cap, after which additional excess loss is translated back to $B$ through a second layer of reinsurance from $C$. Of course, deductibles may reduce the total loss borne by $B$. In a real application, $B$ likely knows little of the network structure outside of first degree connections and so may be unaware that it is also liable for parts of the second layer of reinsurance coverage of $A$.

\begin{figure}
\centering \includegraphics[width=6cm]{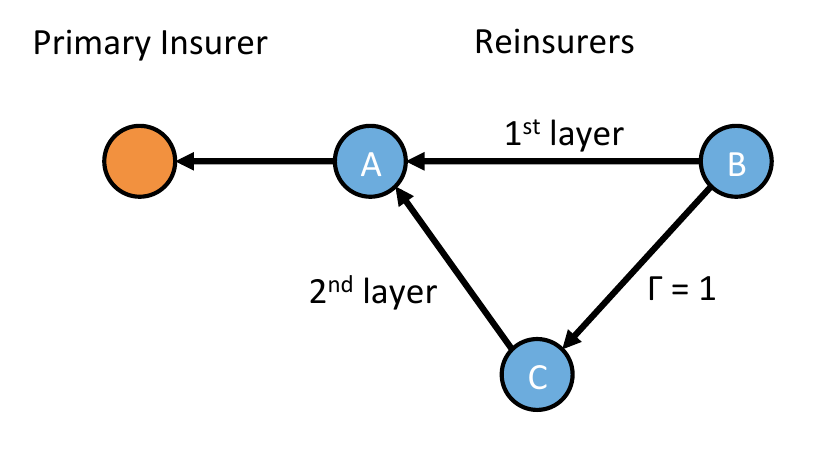}
\caption{Example in which a contract cap is counteracted. If the cap on $B$'s coverage of $A$ is passed, $B$ is still liable for additional coverage of $A$ through $C$'s second layer coverage of $B$.}\label{fig:cap_negated}
\end{figure}


We have demonstrated the emergence of reinsurance spirals and the extreme bearing of reinsurance losses due to network structure using simple examples. While these examples are illustrative, more complicated examples, as in real world reinsurance networks, can exhibit the same effects.

\subsection{Extreme parameter sensitivity}

\paragraph{Incomplete network information.} In a real setting, small groups of firms have incomplete information about the global network structure. They face intrinsic uncertainty of global contract parties and parameters. Indeed, in many cases, the network is unable to be fully observed in principle, even by regulators or the industry as a whole. We have previously mentioned that the parameters of some real industry contracts are not even agreed upon beforehand. Instead, the `contract' is really just an agreement to in the future agree on a contract, and so the parties of the contract do not even know the actual terms of the contract.

\paragraph{Difficulties in measuring risks.}
As a result of this uncertainty, there is also high uncertainty about which dangerous structures can emerge. For example, it is difficult to determine if, by taking a given contract position, a firm exposes itself to being a damping node. Thus firms face high uncertainty about their extreme bearing of reinsurance losses.\footnote{In the unlikely event that nodes have complete information about the network, identifying problematic structures would remain difficult from an algorithmic perspective, as we discussed in the previous sections. As we also noted, we can determine whether a 100\% cycle occurs in a given network using the spectral radius, but this criterion is in general not guaranteed to work for relaxed cycles.}

Additionally, even small perturbations in the network parameters can lead to large differences in losses and where extreme losses are borne, presenting additional complexities to risk management. To illustrate this, consider again the example in Figure~\ref{fig:spiral}. Small changes to the contract caps (e.g., switching which contract has a cap of 10 vs. 11) affect which edge cap is met and, in turn, who bears all excess loss. In the real world, uncertainty around network parameters is likely to be large, which only exacerbates the problem. We demonstrate this high sensitivity to parameters using real network data in the next section.

\paragraph{A new incentive to cooperate on fraud prevention.}
Given the complexity of real world networks, fraud can be quite difficult to uncover. In principle, to fully protect against fraud, a system needs to allow all parties to verify that upstream liabilities are valid propagations of primary insurance claims. Ordinarily, firms only have access to audit the direct claims they receive and must trust all other firms to properly audit their claims and to not collude in this regard. In the industry today, there is typically no good way to detect fraud. Indeed, there is a sense that reinsurers only work with insurers they trust.\footnote{Aside from unequivocal fraud, there are commonly grey instances in which a reinsurer will pay a little more on a contract if the insurer is a `good customer', which can affect claims down the chain. Source: private conversations with an insurance executive.}
	
Due to the parameter sensitivity of these systems, even very small fraud can have outsized effects on the equilibrium. Further, a fraudulent ordering of payments could also affect the equilibrium after taking into account clearing.  As individual firms cannot fully observe the network, they face wild variability in the valid claims they could face given a shock. Each firm should be highly suspect of whether the network of claims is correct, especially if they face high losses under the proposed claims. Even if fraud is very unlikely and only very small, a firm's potential benefit from exposing it could be quite large (e.g., the difference of the firm defaulting or not). Thus our sensitivity results suggest a very strong--and to our knowledge, previously undiscussed--incentive for a large majority of firms to band together to combat fraud, and potentially share information with each other in order to do so.


\subsection{Implications for contract design}

From a global contract design perspective, there are a few things that can help mitigate spirals and parameter sensitivity. We discuss contract deductibles and proportional contracts. However, these are not guaranteed to be effective.

Contract deductibles help to lower the excess loss that is recirculated through graph cycles by absorbing some loss before each edge activation. However, given a large shock, deductibles may not be enough to prevent disproportionate network effects from spirals, as was the case during the Lloyd's and London reinsurance market spiral during the 1980s.

The cap-deductible layering structure of XL reinsurance obfuscates risks and heightens sensitivity to parameters by adding lots of nonlinearities into the network. On top of this, we have to clear the network liabilities given the available equity, which adds additional nonlinearities. On the other hand, a system composed of proportional reinsurance contracts is much simpler to compute as liabilities can be calculated through a linear system. This removes many of the nonlinearities (but not all as we still need to clear the network), which helps make the risk faced by a firm in the network clearer. This lessens the chances that firms think they are adequately reinsured but later find out otherwise. This may also lead to less parameter sensitivity since, as pointed out in a previous section, the liabilities of a proportional system are determined by solving a single linear system. On an aggregate level, this may also lower systemic risk, which we examine using real network data in the next section.

\section{Simulations with Real Network Data}

In this section, we investigate two questions posed in the previous section using simulations on real reinsurance network data:
\begin{itemize}
\item Is there high parameter sensitivity in real reinsurance networks? I.e., is it difficult to estimate the risk faced by a firm in the network from a particular shock? We demonstrated that this is a theoretical issue in the previous section.

\item Are XL or proportional contract systems better from a systemic perspective? In the previous section, we demonstrated that a system based on XL contracts adds many nonlinearities to the system, which can serve to obfuscate risk and concentrate losses. A proportional system, on the other hand, has much less nonlinearities.
\end{itemize}

We also briefly explore the effect of time dependency of claims. Our code is freely available at \url{www.github.com/aklamun/reinsurance_networks}.

\subsection{Network construction}

As the basis for our simulations, we use real network data on property and casualty reinsurance from 2012 Schedule F Part 3, as obtained from the National Association of Insurance Commissioners (NAIC) \citep{naic_reins_data}. This data details premiums ceded to reinsurers by US insurance companies. Naturally, this data does not provide all contract parameters, so we estimate these using common rules of thumb in the insurance industry, which we back up with data where available.

We develop methods for constructing plausible networks of XL and proportional contracts consistent with the data. These methods and the more general simulation setup are detailed in Appendix~\ref{appendix:simulations}. In our simulations, we consider 1-in-100 and 1-in-250 year shocks to the network.

Figure~\ref{fig:network_graph} gives a visualization of a resulting XL network with edges weighted by $\gamma$. The figure shows a core-periphery structure. The core is composed of a central group of reinsurers and primary insurers who reinsure through most of them. This core-periphery structure is common in a variety of financial networks, see e.g. \citep{craig14}.

\begin{figure}
\centering
        \includegraphics[width=0.5\textwidth]{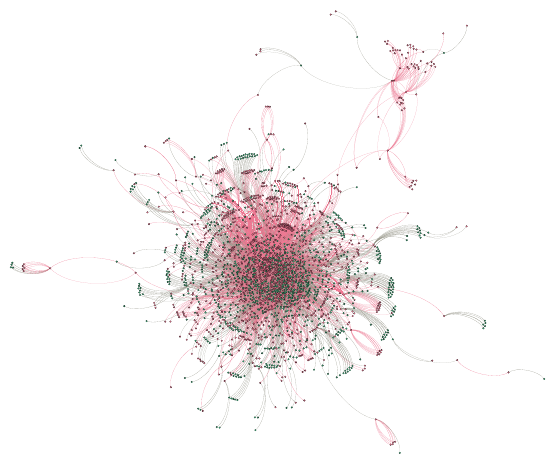}
        \caption{Network visualization. Pink nodes are reinsurers, green nodes are primary insurers.}\label{fig:network_graph}
\end{figure}

\subsection{Sensitivity to parameter perturbations}

In our first set of simulations, we examine the firm-level effects of perturbations in the XL financial network parameters. As discussed in the previous section, the nonlinearities added by layering structure make it difficult to evaluate the exposure of a firm to a shock under parameter uncertainty. Exposures under slightly different parameter sets could be completely different on a theoretical level. In these simulations, we demonstrate that sensitivity occurs in reasonable estimates of a real world reinsurance network.

In these simulations, we construct an XL network from our data, which is our base case for comparison. We then perturb the network parameters by a factor $\delta$ as described in Appendix~\ref{appendix:simulations}. For moderately small values of $\delta$, these perturbations are conservative because firms face a lot of uncertainty about how other firms' contracts are structured and, because of market forces, there are intrinsic uncertainties about each firm's capital available to pay out liabilities at the time of clearing. Additionally, for privately owned insurers, equity values are not publicly available.

With a given shock $sh$, in each simulation, we calculate a liabilities matrix $L$ and clear the liabilities using methods from \citep{eisenberg01}.\footnote{Note that the clearing in \citep{eisenberg01} assumes reinsurance contracts are on the same level of clearing importance as retrocession contracts. By using this clearing, we also assume that limited liability is always invoked between reinsurers. In reality, however, there are vague `parental guarantees' between companies in the same group as well as some degree of joint and several liability, in which regulators can force surviving insurance firms to take on liabilities of failed firms.} This second step is done without default costs for simplicity and outputs a clearing payment vector $p$, representing the total payment from each firm, and a default indicator vector. After the simulation, we calculate the vector of end equities $e_1$ as
$$e_1 = e_0 - p + L^T\alpha - sh,$$
where $e_0$ is the initial equity vector and $\alpha$ is defined component-wise as
$$\alpha_i = \begin{cases}
\frac{p_i}{(L\mathbf{1})_i}, & \text{ if } (L\mathbf{1})_i > 0 \\
0, & \text{ otherwise}
\end{cases},$$
where $\mathbf{1}$ is the all ones vector. The multiplicative equity return is then $e_1/e_0$. A return of 1 represents no loss, 0 represents complete loss of capital, and a negative value means that a primary insurer has outside liabilities that are unable to be covered after clearing. Note that, under this definition, reinsurers face a return floor of zero (i.e., $e_1\geq 0$ since $sh=0$ for a reinsurer); this makes sense because they have limited liability. In a legal sense, limited liability could be applied to primary insurers; however, it will be useful to us to explore the uncovered primary losses that are represented using our definition above for equity. Uncovered primary liabilities represent a failure of the system, as the purpose of the insurance-reinsurance industry is to provide protection on physical infrastructure. This does not happen if primary liabilities are not met.

Under a static 1-in-250 year shock, we run simulations with 2.5\%, 5\%, 10\%, and 20\% perturbations, each with 50 random samples. We examine how firm equity returns and defaults differ between the perturbed systems and the base system. Figure~\ref{fig:compare_equity_perturb} shows the extent to which these perturbations change firm equity returns. These histograms are over maximum differences observed in the 2609 network firms, excluding those with zero observed difference. Even under small 2.5\% perturbations, a firm's equity return can differ by 100 percentage points.\footnote{Note that this high uncertainty in equity return could also occur at smaller perturbation levels. In this study, we don't attempt to numerically find a lower bound.} The tail of the distribution fattens quickly as the perturbation magnitude increases. These results demonstrate that small uncertainties in financial network parameters can lead to wild differences in outputs, as demonstrated with firm equity returns.

\begin{figure}
    \centering
    \begin{subfigure}[b]{0.48\textwidth}
        \includegraphics[width=\textwidth]{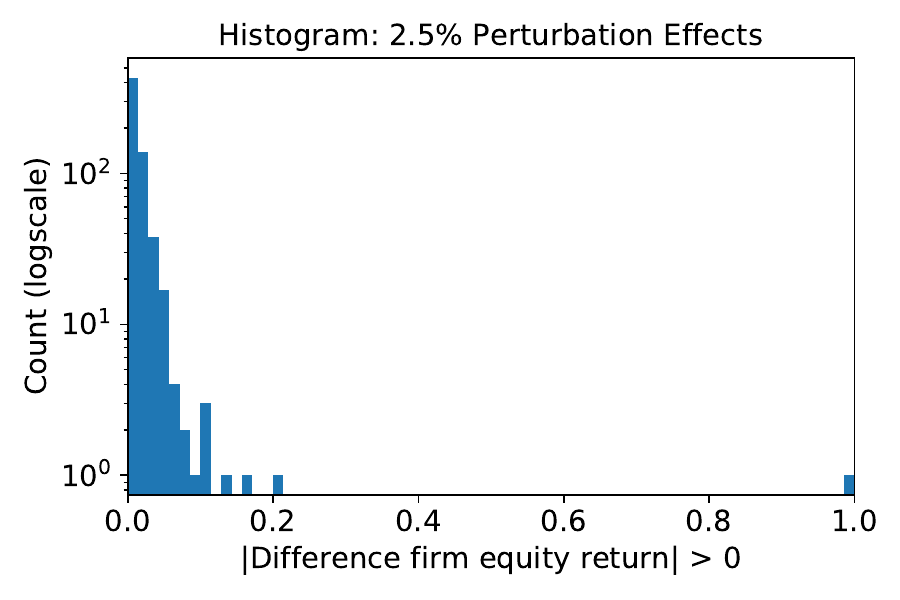}
        \caption{2.5\% perturbation}\label{fig:compare_equity_perturb_2.5}
    \end{subfigure}
    \begin{subfigure}[b]{0.48\textwidth}
        \includegraphics[width=\textwidth]{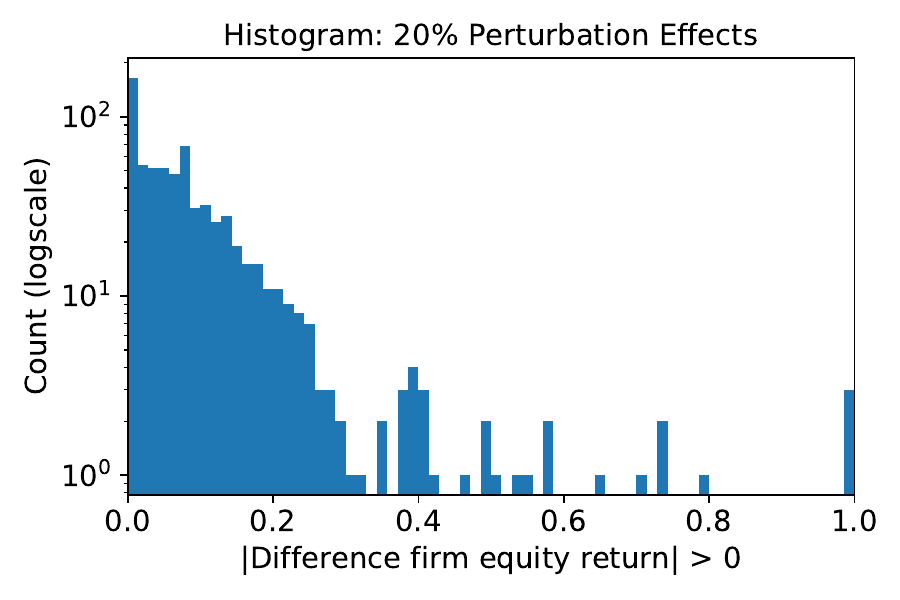}
        \caption{20\% perturbation}\label{fig:compare_equity_perturb_20}
    \end{subfigure}
\caption{Log-scale histograms of perturbation effects on firm equity returns, measured by maximum absolute value of change from the base case to the perturbed case, under a static shock.}\label{fig:compare_equity_perturb}
\end{figure}

These perturbations altered the default status of 8 firms (2.5\% perturbation) up to 21 firms (20\% perturbation). Additionally, these perturbations affected firm equity levels on the order of \$790M (2.5\% perturbation) up to \$110B (20\% perturbation). This demonstrates that perturbations can affect key players of the market. An individual firm could be wildly uncertain about the risks it faces from a given shock given even small network uncertainties.

\subsection{Systemic effects of contract structures}

In our second set of simulations, we examine the systemic effects from different contract structures. As discussed in the previous section, systems of XL contracts can have the effect of concentrating losses in unpredictable ways. This can cause firms to mistakenly think they are properly reinsured when in effect they are not. On the other hand, a system of proportional contracts can have much less obfuscation of risk. In these simulations, we demonstrate that, given real world network structure, proportional reinsurance systems are more stable in the face of tail risk than comparable XL reinsurance systems.

A key word here is of course `comparable'. Our methods for constructing comparable networks is described in Appendix~\ref{appendix:simulations}. Note that this comparison is limited because it is likely that in the real world the equilibrium graph structure, premiums ceded, and firm capital levels would be different between the different settings, whereas we are setting these constant. In the different settings, firms may make different decisions about these parameters. This said, our comparison is still useful because it shows that, given the same aggregate costs in terms of reinsurance premiums ceded and firm capital levels, the reinsurance market could perform systemically better.

We simulate 50 1-in-100 year shocks and 50 1-in-250 year shocks to two comparable XL and proportional reinsurance systems. We compare the aggregate number of defaults, the aggregate uncovered primary claims, and the distribution of firm equity returns for each network in each scenario.

Figure~\ref{fig:compare_agg} compares the number of defaults and the uncovered primary liabilities in each shock scenario. The number of defaults is a common measure of the resilience of financial networks. We argue that the level of uncovered primary liabilities is also a useful comparison because the purpose of the reinsurance industry in the first place is to redistribute risk such that primary insurance liabilities (insurance on real world infrastructure) can be more easily met during shocks. By both measures, the XL system performs consistently worse than the proportional system.

\begin{figure}
    \centering
    \begin{subfigure}[b]{0.48\textwidth}
        \includegraphics[width=\textwidth]{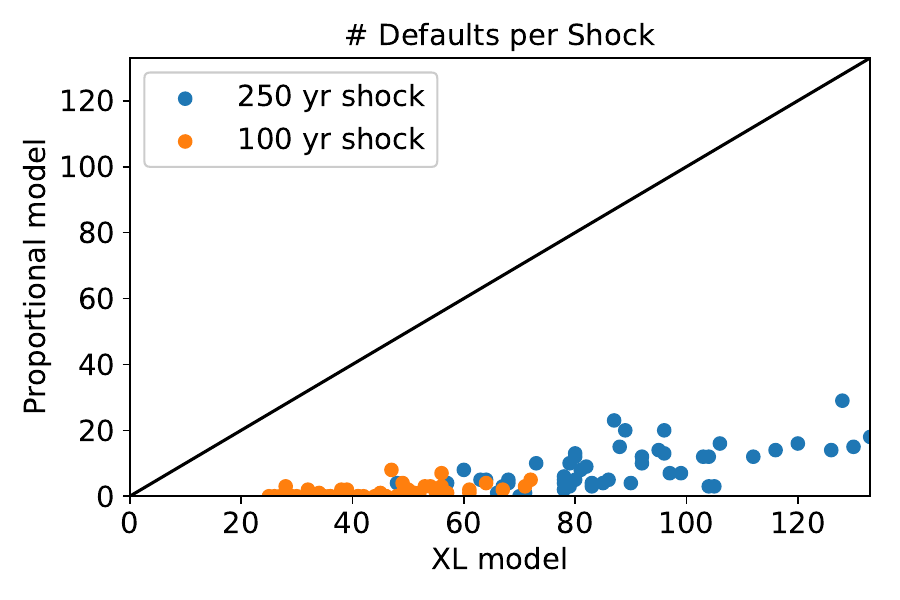}
        \caption{Number of defaults}\label{fig:compare_agg_defaults}
    \end{subfigure}
    \begin{subfigure}[b]{0.48\textwidth}
        \includegraphics[width=\textwidth]{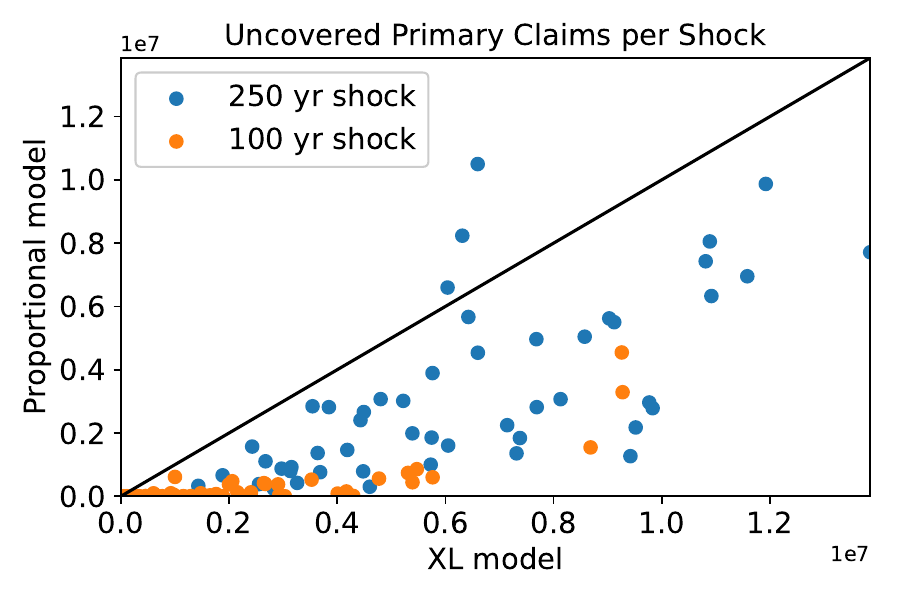}
        \caption{Total uncovered primary liabilities}\label{fig:compare_agg_uncoveredl}
    \end{subfigure}
    \caption{Aggregate comparison of proportional vs. XL systems. Each point is a shock realization.}\label{fig:compare_agg}
\end{figure}

Figure~\ref{fig:compare_equity_hists} shows histograms comparing firm equity returns between the XL and proportional systems under 1-in-100 year and 1-in-250 year shocks. These histograms collapse two (notably dependent) dimensions of data into one: the 2609 firms in the network and the 100 shock simulations. Thus each firm accounts for 100 data points. These histograms help to visualize the firm-level effects across the simulations.

Figures~\ref{fig:hist_firm_equity_xl} and \ref{fig:hist_firm_equity_prop} show empirical distributions of firm equity returns, which demonstrate that firms face higher tail risk of losses under the XL system than the proportional system. Note that the histogram spikes at 0 are caused by the limited liability of reinsurance companies. As discussed in the previous subsection, we could apply the same limited liability to primary insurers, but find it more useful to represent the uncovered primary liabilities within the distribution.

\begin{figure}
    \centering
    \begin{subfigure}[b]{0.48\textwidth}
        \includegraphics[width=\textwidth]{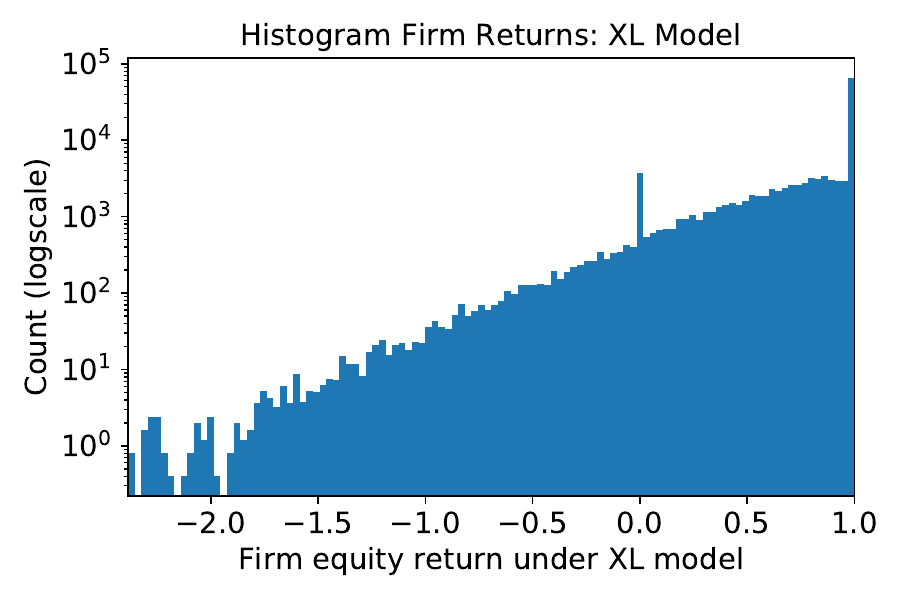}
        \caption{Firm equity returns under XL system}\label{fig:hist_firm_equity_xl}
    \end{subfigure}
    \begin{subfigure}[b]{0.48\textwidth}
        \includegraphics[width=\textwidth]{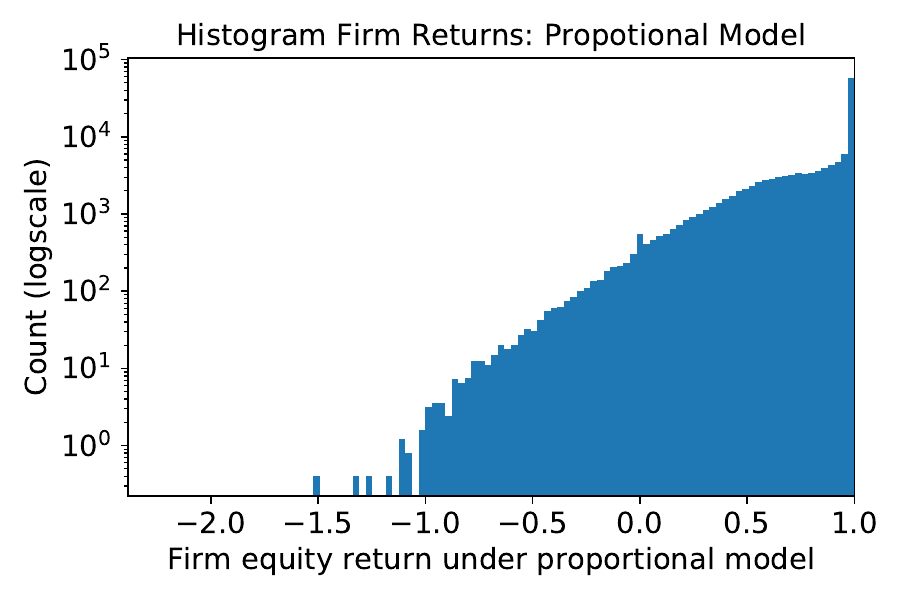}
        \caption{Firm equity returns under proportional system}\label{fig:hist_firm_equity_prop}
    \end{subfigure}
	\begin{subfigure}[b]{0.48\textwidth}
		\includegraphics[width=\textwidth]{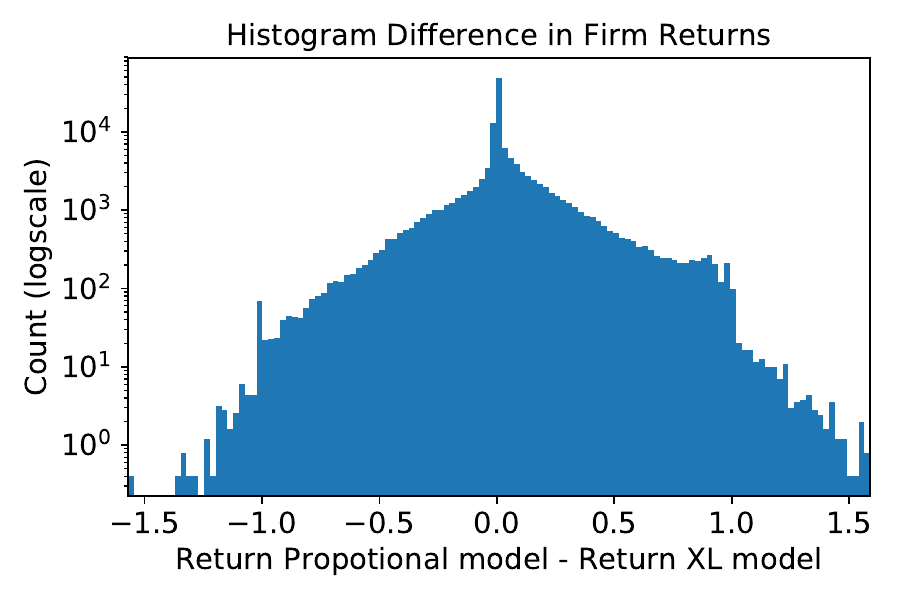}
		\caption{Difference in firm equity returns between models}\label{fig:hist_firm_equity_diff}
	\end{subfigure}
	\begin{subfigure}[b]{0.48\textwidth}
		\includegraphics[width=\textwidth]{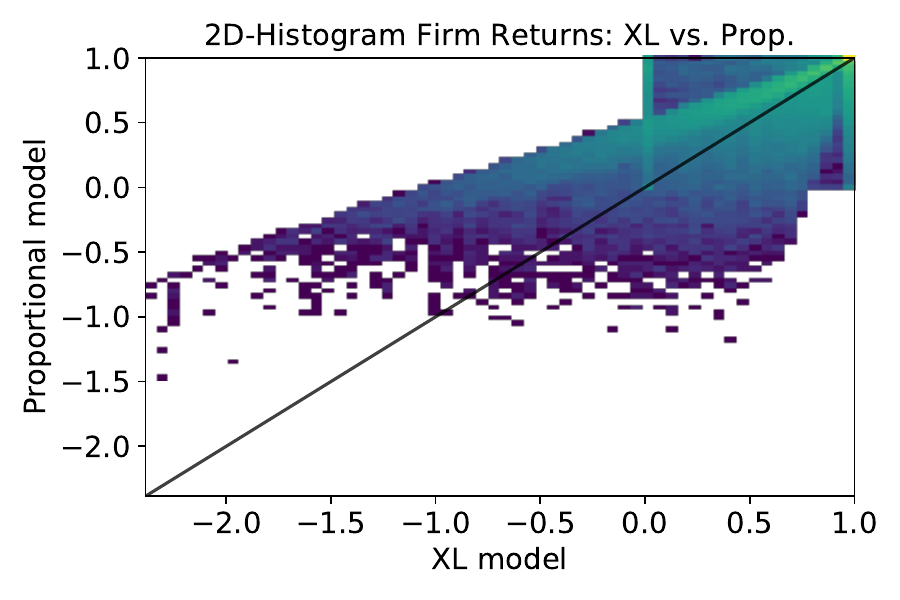}
		\caption{2D histogram returns in XL vs. proportional}\label{fig:histogram2d_compare_firm_equity}
	\end{subfigure}
    \caption{Firm equity returns under XL and proportional systems. Histograms include all simulations, weighted by relative probability of 1-in-100 vs. 1-in-250 year events (60\% vs. 40\%).}\label{fig:compare_equity_hists}
\end{figure}

 Figure~\ref{fig:hist_firm_equity_diff} shows the histogram of firm-level differences in returns between the two models, taking into account all scenarios. In these scenarios, $>49\%$ of firms are better off under the proportional model than the XL model. In particular, the additional cost of the proportional structure to the average firm is small. Consistent with the previous histograms, however, firms are predominantly better off in the tails under the proportional model than the XL model. Figure~\ref{fig:histogram2d_compare_firm_equity} shows a 2D histogram of firm equity return under the proportional model vs. the XL model. The same tail structure can be seen in the 2D histogram. Note that again we can see the limited liability effect of reinsurance companies in the square structure around 0 in the 2D histogram. Reinsurers' returns are constrained to the square between 1 and 0 in the 2D histogram due to the limited liability, whereas the primary insurers' returns form the triangle structure.

\subsection{Effects of time dependency of claims}

We now briefly explore the effects of claims that come in over multiple time periods. In reality, claims take varying time and up to several years to reach reinsurers. Two main factors contribute to this. First, claims are largely manually reported from one party to the next. This paperwork process can take considerable time to trickle through chains of reinsurers in the network. Reinsurers who are closer to the source see claims earlier. Second, some claims do not materialize until future years. Most reinsurance contracts are tied to property and casualty, in which losses commonly develop over 5-10 years. Loss run-offs beyond 10 years are common with major shocks like terrorist attacks and hurricanes. For instance, claims from the World Trade Center attacks were litigated over nearly a decade. Further, asbestos claims can come in decades after the fact.

We consider a simple setup in which claims come in over two time periods. After the first claims come in, the network liabilities are cleared and the firm capital is updated. The network is cleared again after the second claims come in. Such a setting, with multiple rounds of clearing underlies the  emerging literature on dynamic network models \cite{feinstein18b, kusnetsov19, capponi15}. We compare this to a single period setup in which all of the claims come in at the same time.

Formally, the network faces shocks $sh_1$ and $sh_2$ in periods one and two respectively. These yield network liabilities matrices $L_1$ and $L_2$ respectively. We apply clearing from \citep{eisenberg01} sequentially over the periods. In period one, liabilities $L_1$ are cleared and firm capital is updated as follows. Given starting equities $e_0$, shock $sh_1$, and liabilities $L_1$, the period one clearing payment vector $p_1$ is calculated. Then period one end equities are calculated as
$$e_1 = e_0 - p_1 + L_1^T\alpha_1 - sh_1,$$
where $\alpha_t$ is componentwise $\alpha_{t,i} = \begin{cases}
\frac{p_{t,i}}{(L_t\mathbf{1})_i}, & \text{ if } (L_t\mathbf{1})_i > 0 \\
0, & \text{ otherwise}
\end{cases}$ and $\mathbf{1}$ is the all ones vector.

Recall that as the shocks only directly affect the primary insurers, the minimum equity of a reinsurer is zero after the clearing. As before, we allow negative equities for primary insurers to account for uncovered primary losses. As primary insurers are simply leaves in the network, these do not affect the next period clearing.

In period two, remaining capital $e_1$ is used to clear liabilities $L_2$. In particular, if a reinsurer defaults in period one, any retrocession payments they receive in period two are channeled to the period two liabilities that triggered these payments. Given equities $e_1$, shock $sh_2$, and liabilities $L_2$, the period two clearing payment vector $p_2$ is calculated. Then period two end equities are calculated as
$$e_2 = e_1 - p_2 + L_2^T\alpha_2 - sh_2.$$
The multiplicative equity return across the two periods is then $e_2/e_0$.

The effect of the two period clearing is that insurers who are connected to defaulted reinsurers may have their claims paid in different fractions of face value depending on whether they are paid in the first or second period. This compares to a single period clearing, in which everyone is paid out in the same proportion. This means that earlier claims may effectively get seniority over later claims. A reinsurer may be able to pay out in full for early claims, but may enter default in the second period. However, the opposite can also happen. A reinsurer may be unable to pay in full on early claims, but is able to pay more if retrocession contracts are activated when the second wave of claims comes in. Notice, however, that a firm that defaulted in the first period will never have excess value in the second period because reinsurance coverage is $<100\%$.

We simulate 50 1-in-250 year shocks that are split between two periods. For each of the 50 simulations, we generate two random 250-year shocks distributed in the same way as described in Appendix~\ref{appendix:simulations} (i.e., not uniformly distributed, but proportional to the size of premiums). We divide the magnitude of each by two; these are the period one and period two shocks. We then compare this to the single aggregate shock.

Figure~\ref{fig:compare_equity_time_diff} shows the histogram of firm-level differences in returns between  two periods and one period clearing. The histogram collapses two (notably dependent) dimensions of data into one: the $2609$ firms in the network and the 50 shock simulations. Thus each firm accounts for $50$ data points. This histogram is useful for visualizing how much the time dependence of claims can affect firms' equities. Note that the average difference is very close to 0 (0.00032) although the distribution as visualized is asymmetric. This reflects that the structural change is not adding excess costs but rather changing the distribution of costs. To illustrate how much firm equities can change between the two-period and one-period clearing schemes, in every simulation, at least $5\%$ of firms saw $>22\%$ absolute change and at least $1\%$ of firms saw $>60\%$ absolute change.

\begin{figure}
	\centering
	\includegraphics[width=8cm]{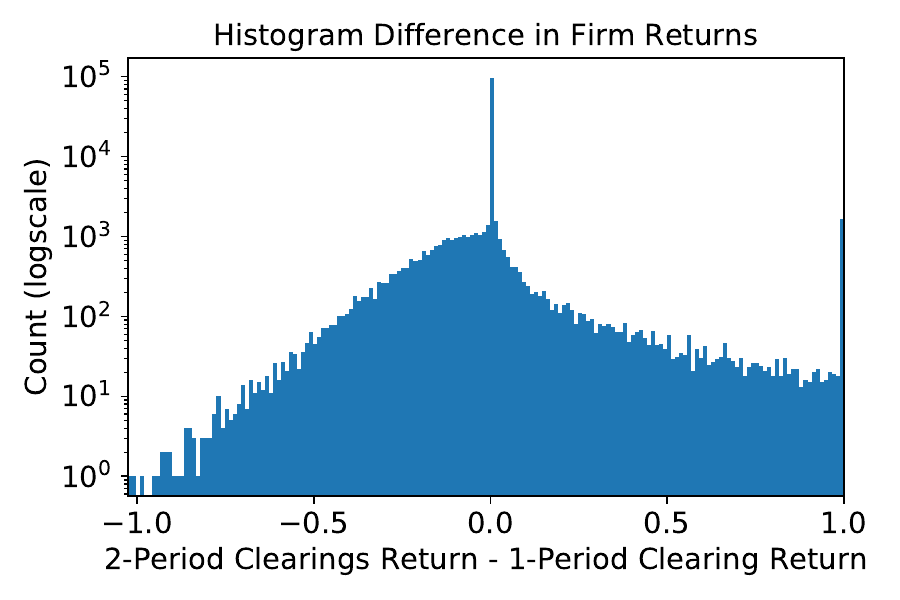}
	\caption{Difference in firm equity returns between two-period and one-period clearing.}\label{fig:compare_equity_time_diff}
\end{figure}





\section{Concluding Remarks}

Current reinsurance risk models do not capture network effects, which we show can be quite extreme. We have demonstrated that even if firms know the global network structure unreasonably well (i.e., with small uncertainty), they can be wildly uncertain about how losses will be distributed from a given shock. Tail risk from network structure should be taken into account in determining capital requirements, evaluating counterparty risk, and pricing reinsurance contracts. This exposes inherent inefficiencies in the current system.

\paragraph{Combating fraud.}
Our sensitivity results reveal a strong new incentive for firms to band together to combat fraud. This could be achieved through a trusted central party or, in the absence of such a party, a distributed ledger system. A blockchain could provide guarantees on fraud prevention if it is able to record audited and time-ordered claims and reinsurance payments. Such a system would need to allow participants in the network to independently verify the equilibrium liabilities and clearing state in the network. We propose future work to design such a system that works given the real world constraints around contract privacy. We note, however, that the strong incentive we have demonstrated may make firms more willing to share some data to contribute to fraud prevention, thus relaxing these contraints. Many organizations are making a concerted effort to incorporate blockchain systems into the (re)insurance industry. Our paper helps to inform them about the problems they should be addressing.

\paragraph{Measuring risk.}
We have revealed dangerous structures that lead to tail risk from network effects. We propose new tools that better measure and classify tail risk of positions (e.g., nodes or contracts) within the network. One approach is to use inner Monte Carlo simulations over a range of shocks, and outer Monte Carlo simulations over a range of parameters; however, convergence may be costly. A second approach is to use machine learning classification algorithms. This would entail generating a wide variety of graphs and parameters and evaluating losses from contagion shocks under different scenarios via algorithms from our paper. Using this as a training set, the aim is to detect graph structures that predict which nodes bear tail risk from spirals. We believe the structure that we have revealed about tail risk in this paper will aid in the construction of such methods. One promising result from our simulations is that, under a given shock, the losses of many network nodes appear robust to parameter error whereas others suffer more chaotic behavior. This suggests that classification algorithms may be successful in predicting which nodes bear high uncertainty and are therefore more susceptible to model error.

\paragraph{Designing better systems.}
Our simulations suggest that, for the same societal-level costs in terms of reinsurance premiums and capital locked in the reinsurance industry, the industry can be better structured to perform its social purpose more effectively during extreme events. We would like to extend this to a market design perspective. One issue is that, in isolation, firms can have an incentive to require caps on payouts (although this is not clear), but at the network level this does not appear optimal. We leave it to future work to explore the robust systemic design perspective taking into account how the $\gamma$ matrix changes with respect to changes in contract structure.

In recent years, catastrophe bonds held by nontraditional players, such as hedge funds, have become more popular in place of traditional reinsurance contracts. One advantage of these is that they could become additional dampers in the system as they absorb losses in the system without recirculating them. This has the tradeoff, however, of more significantly interconnecting the larger financial system, which can cause other potential exposures that are relevant from a market design perspective.

Lastly, stemming from our discussions with industry executives, we propose the following extensions to the model and analysis.
\begin{itemize}
	\item \textbf{Time dimension:} As discussed in the previous section, time dependence of claims can have a large effect on firms' equities. This warrants further work. We note that NAIC provides an extensive historical database on insurance loss run-offs in Schedule P.
	
	\item \textbf{Liquidity factor:} Extreme events in the insurance-reinsurance industry--such as high concentration of large losses due to network structure--could trigger a liquidity crisis from fire sales of risky assets. This can amplify losses within and beyond the reinsurance industry, propagating an insurance-specific event into a systemic crisis.
\end{itemize}


\bibliographystyle{apalike}

\newpage

\normalsize
\appendix
\section{Proofs}\label{appendix:proofs}

\paragraph{Lemma~\ref{lemma:BC-constant_sets}} \hypertarget{pf:BC-constant_sets}{}
\begin{proof}
	For every $\ell \geq 0$, there is a unique corresponding $B$ and $C$ defined (and so also a unique $\Psi$). Note that on the boundary between $(B,C)$-constant sets multiple $B$s and $C$s could be defined equivalently. This is because $\Phi$ is an intersection of linear systems on the boundaries--the difference in possible $B$s and $C$s comes from edges that have exactly met their deductible or cap respectively but have no excess liability under $\ell$. In these cases, $B_{ii}$ or $C_{ii}$ respectively can be set equivalently to 1 or 0, but a unique selection is defined in the definition.
	
	The derivatives of $B(\ell)$ and $C(\ell)$ are defined and zero except at points of discontinuity since $B$ and $C$ only change at thresholds. $B$ and $C$ are defined such that their value on the boundary (i.e., at points of non-differentiability in either $B$ or $C$) are constant with a value on one side of the boundary. As there are $2^m$ possible $B$ and $C$ matrices (1 or 0 for each diagonal entry), the $(B,C)$-constant sets form a finite partition of $\{\ell \vert \ell \geq 0\}$ of size at most $2^{2^m}$.
	
	To establish convexity, notice that systems of linear inequalities define $B$ and $C$. And so the $(B,C)$-constant sets, each the intersection of the pre-image of a given $B$ and $C$ values, are convex sets since they are the intersections of half-spaces from each inequality.
\end{proof}

\noindent\rule{\textwidth}{1pt}

\paragraph{Theorem~\ref{prop:unique}} \hypertarget{pf:unique}{}
\begin{proof}
	Lemma~\ref{lemma:bertsekas} gives us that a linear system with matrix $\Omega$ is a contraction with respect to some weighted Euclidean norm. Let $\| \cdot \|_s$ be such a norm and let $\alpha\in [0,1)$ be the corresponding Lipschitz constant. Since our matrices are non-negative, the Perron-Frobenius theorem gives us that for $(B,C)\in \mathcal{K}$, $\rho\Big((I-C)\gamma BX(I-C)\Big) \leq \rho(\Omega) < 1$.
	
	Note that the derivatives of $B(\ell)$ and $C(\ell)$ are zero except at points of discontinuity. On the subsets of the domain space on which $(B,C)$ is constant ($(B,C)$-constant sets), $\Phi$ is a linear system described by $(I-C)\gamma BX(I-C)$. This can be written as
	$$\begin{aligned}
	\Phi(\ell) &= (I-C)\gamma B \Big( X(I-C)\ell + XCc + (I-C)(s-d)\Big) + Cc \\
	&= \Psi^T \tilde \gamma \tilde B (\tilde X \Psi \ell + \tilde v) + \bar \ell.
	\end{aligned}$$
	
	Let $\ell_1, \ell_2$ be points in a $(B,C)$-constant set. Then
	$$\begin{aligned}
	\|\Phi(\ell_1) - \Phi(\ell_2) \|_s &= \|\Psi^T \tilde \gamma \tilde B (\tilde X \Psi \ell_1 + \tilde v) + \bar \ell_1 - \Psi^T \tilde \gamma \tilde B (\tilde X \Psi \ell_2 + \tilde v) + \bar \ell_2\|_s \\
	&= \| \Psi^T \tilde \gamma \tilde B \tilde X \Psi(\ell_1-\ell_2)\|_s \\
	&\leq \|\Omega(\ell_1-\ell_2)\|_s \\
	&\leq \alpha \|\ell_1-\ell_2\|_s,
	\end{aligned}$$
	where the third line follows because $0 \leq \Psi^T \tilde \gamma \tilde B \tilde X \Psi \leq \Omega$ element-wise. 
	Thus $\Phi$ is a contraction with respect to $\|\cdot\|_s$ locally on each $(B,C)$-constant set.
	
	Note that for $\hat \ell$ on the boundary of a $(B_1,C_1)$-constant set and a $(B_2,C_2)$-constant set,
	$$\Psi_1^T \tilde \gamma_1 \tilde B_1 (\tilde X_1 \Psi_1 \ell + \tilde v_1) + C_1 c = \Psi_2^T \tilde \gamma_2 \tilde B_2 (\tilde X_2 \Psi \ell + \tilde v_2) + C_2 c$$
	since $\Phi$ is continuous. The explanation for this is that, on the boundary, multiple $B$s and $C$s could be defined equivalently. $\Phi$ is an intersection of linear systems on the boundaries. The difference in possible $B$s comes from edges that have exactly met their deductible but have no excess liability. The difference in possible $C$s comes from edges that have exactly met their cap. In these cases, $B_{ii}$ (respectively $C_{ii}$) can be set equivalently to 1 or 0. Hence, the contraction relation extends to the boundaries of $B$-constant sets.
	
	We next show that the contraction relation extends to the union of two adjacent $(B,C)$-constant sets. Choose $\ell_1 \in (B_1,C_1)$-constant set and $\ell_2 \in (B_2,C_2)$-constant set such the shortest path only crosses one $(B,C)$-constant boundary. Since $\|\cdot \|_s$ is a weighted Euclidean norm, there is a shortest path between $\ell_1$ and $\ell_2$ that crosses the boundary between $(B_1,C_1)$- and $(B_2,C_2)$-constant sets. Let $\hat \ell$ be the crossing point of this boundary. Then
	$$\begin{aligned}
	\| \Phi(\ell_1) - \Phi(\ell_2)\|_s &= \| \Phi(\ell_1) - \Phi(\hat \ell) + \Phi(\hat \ell) - \Phi(\ell_2)\|_s \\
	&= \| \Psi_1^T \tilde \gamma_1 \tilde B_1 \tilde X_1 \Psi_1 (\ell_1 -\hat \ell) + \Psi_2^T \tilde \gamma_2 \tilde B_2 \tilde X_2 \Psi_2 (\hat \ell - \ell_2)\|_s \\
	&\leq \| \Psi_1^T \tilde \gamma_1 \tilde B_1 \tilde X_1 \Psi_1 (\ell_1 -\hat \ell) \|_s + \| \Psi_2^T \tilde \gamma_2 \tilde B_2 \tilde X_2 \Psi_2 (\hat \ell - \ell_2)\|_s \\
	&\leq \|\Omega(\ell_1 - \hat \ell)\|_s + \|\Omega(\hat \ell - \ell_2)\|_s \\
	&\leq \alpha \|\ell_1 - \hat \ell\|_s + \alpha \| \hat \ell - \ell_2\|_s \\
	&= \alpha \|\ell_1 - \ell_2\|_s, 
	\end{aligned}$$
	where the second line follows since either $(B_1,C_1)$ or $(B_2,C_2)$ can be used in $\Phi$ along the boundary, the third line follows from the triangle inequality, the fifth line follows from the contraction relation on $(B,C)$-constant sets and their boundaries, and the sixth line follows since $\hat \ell$ is on the shortest path from $\ell_1$ to $\ell_2$.
	
	Next, consider the shortest path (a line) between any two points in the space $\{ \ell \vert \ell \geq 0\}$. As established by Lemma~\ref{lemma:BC-constant_sets}, the $(B,C)$-constant sets are convex, which means that a line cannot cross the boundary of any $(B,C)$-constant set more than twice. Thus, the shortest path between the points can only cross finitely many boundaries (at most $2\cdot 2^{2^m}$, or two for each possible $(B,C)$-constant set). Then, by induction on the number of $(B,C)$-constant sets along the shortest path, the contraction relation of $\Phi$ extends to the union of all $(B,C)$-constant sets, which is equivalently the whole space $\{ \ell \vert \ell \geq 0\}$ by Lemma~\ref{lemma:BC-constant_sets}.
	
	We now need to show that solutions are restricted to a compact set. Since $\rho(\Omega) < 1$, we can derive an upper bound for the solution by solving the dominating linear system $\Omega$ (which may or may not come from a feasible $(B,C)\in \mathcal{K}$), taking the maximum coordinate, and forming the hypercube in which coordinates are bounded by $0$ and this maximum coordinate. The Banach fixed point theorem then gives the result.
\end{proof}

\noindent\rule{\textwidth}{1pt}

\paragraph{Theorem~\ref{prop:kleene}} \hypertarget{pf:kleene}{}
We will first introduce the machinery behind the Kleene fixed point theorem following the exposition from \citep{baranga91} and then use it to prove Theorem~\ref{prop:kleene}.

Let $(P, \leq)$ be a partially ordered set, meaning the binary relation $\leq$ is reflexive, antisymmetric, and transitive.
\begin{itemize}
	\item $(P, \leq)$ is \textbf{$\omega$-complete} if every increasing (i.e., nondecreasing) sequence $\{x_n\}_{n \in \mathbb{N}}$ in $P$ has supremum in $P$.
	\item A function $f: P \rightarrow P$ is \textbf{$\omega$-continuous} if it preserves supremums of increasing sequences. I.e., for every increasing sequence $\{x_n\}_{n\in \mathbb{N}}$ in $P$ that has supremum in $P$, the sequence $\{f(x_n)\}_{n\in \mathbb{N}}$ also has supremum in $P$ and
	$$\lim_{n\rightarrow \infty} f(x_n) = f\left( \lim_{n\rightarrow \infty} x_n \right).$$
\end{itemize}
Notice that a $\omega$-continuous function is monotone increasing. This is a direct consequence of preserving suprema of all increasing sequences.

\begin{theorem}
	\textbf{(Kleene fixed point theorem)} Let $(P,\leq)$ be a $\omega$-complete partially ordered set and $f:P\rightarrow P$ be a $\omega$-continuous function. If there is $x\in P$ such that $x \leq f(x)$, then $\bar x~=~\sup \{ f^n(x) \vert n \in \mathbb{N}\}$ is the least fixed point of $f$ in $\{y\in P \vert y \geq x\}$.
\end{theorem}

Let $\bar{\mathbb{R}}$ be the completion of the real numbers with $\infty$. We will work in this extended space and draw our results back to the normal real space. We now prove Theorem~\ref{prop:kleene}.

\begin{proof}
	First notice $(\{\ell \in \bar{\mathbb{R}}^m \vert \ell \geq 0\}, \leq)$ is a $\omega$-complete partial ordering. Choose $x=0$ and note that $0 \leq \Phi(0)$. Notice that we are working in an extension of $\mathbb{R}^m$, and so we may find that the fixed point promised by the theorem is infinite. To address this, we have assumed that there is a (finite) fixed point on $\mathbb{R}^m$, and so the minimum fixed point must also be finite. It now remains to be shown that $\Phi$ is $\omega$-continuous.
	
	Take two sequences $x_n \uparrow \bar x$ and $y_n \uparrow \bar x$ in the partial ordering. We need to establish that $\lim_{n\rightarrow \infty} \Phi(x_n) = \lim_{n\rightarrow \infty} \Phi(y_n)$. This result is immediate if all coordinates of $\bar x$ are finite since $\Phi$ is continuous and monotone increasing. So suppose some coordinates of $\bar x$ are infinite. As we go along the process $\Phi(x_n)$, a finite number of edges and caps can be activated, after which activations stop. Thus there is a step $N$ after which $\Phi$ will be a linear map on the remaining $x_n$s. The same is true for some step $M$ for the sequence of $y_n$s. Then  for $n\geq \max(M,N)$, the $\Phi(x_n)$ and $\Phi(y_n)$ will lie on an increasing hyperplane, with $\Phi(x_n) = \Psi^T \tilde \gamma \tilde B (\tilde X \Psi x_n + \tilde v) + Cc$  and  $\Phi(y_n) = \Psi^T \tilde \gamma \tilde B (\tilde X\Psi y_n + \tilde v) + Cc$, for some $B,C,\Psi(C)$. Since $x_n \geq 0$ we then have 
	$$\begin{aligned}
	\lim_{n\rightarrow \infty} \Phi(x_n) &= \Psi^T  \tilde \gamma \tilde B \tilde X \Psi \lim_{n\rightarrow \infty}x_n + \Psi^T \tilde \gamma \tilde B \tilde v + Cc \\
	&= \Psi^T  \tilde \gamma \tilde B \tilde X \Psi \bar x + \Psi^T \tilde \gamma \tilde B \tilde v + Cc.
	\end{aligned}$$
	The last equality holds since $\bar x$ is the supremum of $x_n$ and thus lies on the same extended hyperplane. The same equality holds for the $y_n$ sequence.
	
	Now define $\Phi(\bar x) := \lim_{n\rightarrow \infty} \Phi(x_n)$ for any sequence $x_n \uparrow \bar x$. By the above, this is well-defined because the value is independent of the sequence chosen. Thus $\Phi$ is $\omega$-continuous. Then the Kleene fixed point theorem gives the results.
\end{proof}

\noindent\rule{\textwidth}{1pt}

\paragraph{Theorem~\ref{prop:tarski}} \hypertarget{pf:tarski}{}
\begin{proof}
	Because $\Phi$ is the composition of an increasing affine map, an element-wise maximum with 0, and an element-wise minimum with $c > 0$, $\Phi$ is non-negative and monotone increasing.
	
	We now show that we can restrict the domain of $\Phi$ to a complete lattice containing all fixed points. In the worst case, all finite caps are met, leaving us with the system $\Psi_0 \gamma X \Psi_0^T$. Since this has spectral radius $<1$, this subsystem has a unique fixed point $\ell_{max}$ by the result in the previous section. Thus, in the worst case, this is the maximum fixed point of $\Phi$. Note that this is dependent on the shock $s$, but such a point exists for each $s$. Let $y$ be the maximum element of $p$ and form the complete lattice $[0,y] \subset \mathbb{R}^m$ bounded in each coordinate by $0$ and $y$.
	
	Restrict the domain of $\Phi$ to $[0,y]$. Then the Tarski fixed point theorem gives us the existence of least and greatest fixed points.
\end{proof}

\noindent\rule{\textwidth}{1pt}

\paragraph{Lemma~\ref{lemma:net_liabilities_1}} \hypertarget{pf:net_liabilities_1}{}
\begin{proof}
Let $L_{i*} := (Le)_i$ and $L_{*i} := (L^T e)_i$. Then $\Delta_i(L) = L_{*i} - L_{i*}$. Note that $L_{*i} = f^i(L_{i*})$, where
$$f^i(L_{i*}) := \sum_j \Big( \Gamma_{ji}\Big((L_{i*} + sh_i - DD_{ji})\vee 0\Big) \wedge CP_{ji}\Big).$$
This is because the amount that reinsurers reimburse $i$ is dependent on the liabilities that $i$ directly faces--i.e., $L_{i*}$.

Then $\Delta_i(L) = f^i(L_{i*}) - L_{i*}$ is monotone decreasing (i.e., nonincreasing) in $L_{i*}$ since reinsurance is limited to 100\%. When a contract deductible is reached, the negative slope lessens. When a contract cap is reached, the negative slope steepens. However, the 100\% reinsurance limit means that the slope is never greater than zero.

Since $L \geq L'$, we also have $L_{i*} \geq L'_{i*}$. The result then follows from the fact that $\Delta_i(L) = \Delta_i(L_{i*})$ is monotone decreasing in $L_{i*}$.
\end{proof}

\noindent\rule{\textwidth}{1pt}

\paragraph{Lemma~\ref{lemma:net_liabilities_2}} \hypertarget{pf:net_liabilities_2}{}
\begin{proof}
	$$\begin{aligned}
	\sum_i \Delta_i(L) &= \sum_i \Big( \sum_j L_{ji} - \sum_j L_{ij}\Big) \\
	&= \sum_{i,j} L_{ij} - \sum_{i,j} L_{ji} \\
	&= 0
	\end{aligned}$$
\end{proof}

\noindent\rule{\textwidth}{1pt}

\paragraph{Theorem~\ref{prop:net_liabilities_equal}} \hypertarget{pf:net_liabilities_equal}{}
\begin{proof}
	Without loss of generality, assume $L' = L^-$, the least fixed point. Since $L\geq L'$,  Lemma~\ref{lemma:net_liabilities_1} implies that $\Delta(L) \leq \Delta(L')$.
	
	Now suppose there exists $i$ such that $\Delta_i(L)<\Delta_i(L')$. Then we in turn have
	$$\sum_j \Delta_j(L) < \sum_j \Delta_j(L').$$
	However, by Lemma~\ref{lemma:net_liabilities_2}, we know that
	$$\sum_j \Delta_j(L) = 0 = \sum_j \Delta_j(L').$$
	Thus there can be no such $i$.
\end{proof}

\noindent\rule{\textwidth}{1pt}

\paragraph{Proposition~\ref{prop:no_caps_iter_converge}} \hypertarget{pf:no_caps_iter_converge}{}
\begin{proof}
We first show that, at each step, the system $(I-\gamma BX)$ is nonsingular. We are given $\rho(\gamma X)<1$. Then, as noted in the proof to Theorem~\ref{prop:unique}, the spectral radii obey
$$\rho(\gamma B X) \leq \rho(\gamma X) < 1$$
for any diagonal $B$ with 1-0 entries since $B$ only serves to remove edges from the initial line graph. Then the Neumann series gives us that $(I-\gamma BX)$ is invertible at each step in the algorithm.

The algorithm converges to the correct solution by a simple monotonicity argument. At each step, we have $B_t \leq \hat B$, where $\hat B$ is the true set of edge activations, since we start with all edges unactivated and edges that become activated are direct propagations of the claims on primary insurers. The sequence of $B_t$ is monotonically increasing in entries since the activation of edges can only increase the number of other edges that become activated. $B_t$ can update at most $m$ times as that is how many edges can become activated. Eventually, we reach a state that represents the correct edge activations, after which the contagion spreads to no further edges, and the edge liabilities are the solution to the resulting linear system. This equilibrium point is the unique fixed point since solving the linear system and checking that $B$ does not change is equivalent to verifying a fixed point of $\Phi$. As each step requires solving a linear system (requiring in general $O(m^3)$ time), and there are at most $m$ steps, the total running time is at most $O(m^4)$.
\end{proof}

\noindent\rule{\textwidth}{1pt}

\paragraph{Proposition~\ref{prop:caps_iter_converge}} \hypertarget{pf:caps_iter_converge}{}
\begin{proof}
	As before, we first show that, at each step, the system $(I-\tilde \gamma \tilde B \tilde X)$ is nonsingular. We are given $\rho(\Omega)<1$. Then, as noted in the proof to Theorem~\ref{prop:unique}, the spectral radii obey
	$$\rho(\tilde \gamma \tilde B \tilde X) \leq \rho(\Omega) < 1,$$
	for any $(B,C) \in \mathcal{K}$ and $\Psi(C)$ since $\tilde \gamma \tilde B \tilde X$ is effectively a subgraph of $\Omega$ after removing edges under $B$ and nodes under $\Psi$.
	
	
	We are given $\rho(\Psi \gamma X \Psi)<1$. This applies to the first iteration of the algorithm. All subsequent iterations involve $(C,B)\in \mathcal{K}$. In particular, the last iterative $\ell$ value at that point in the algorithm is feasible for the given $(C,B)$. Thus at each iteration, we have $\rho(\tilde \gamma \tilde B \tilde X)<1$. Then the Neumann series gives us that $(I-\tilde \gamma \tilde B \tilde X)$ is invertible at each step in the algorithm.
	
	The algorithm converges to the correct solution by a monotonicity argument as in Proposition~\ref{prop:no_caps_iter_converge}. However, the setup here is more nuanced. If the iteration had started at 0 in this setting, we would lose the property that $B_t \leq \hat B$ and $C_t \leq \hat C$, where $\hat B$ and $\hat C$ are the true edge activations and cap activations, as some edge activations could cause the linear solver to attribute more liability to some edges than are allowed by their capacities. While the overcapacity would be corrected in the following iteration, the overcapacity leakage could have caused new activations in $B_t$ that cannot be corrected by the next iteration.
	
	Instead of starting at 0, we start at an upper bound to the solution. Such an upper bound is constructed by assuming all edges are activated ($B=I$) and all finite caps are activated and solving the linear system, which has a unique solution since as shown above. Now we will have $B_t \geq \hat B$ and $C_t \geq \hat C$ at each step since we start with an element-wise overestimate in $B$ and $C$ and any caps or edges that become deactivated through this process will have been unsupported given the overestimate. In this event, either we correct an element in $C$ downward or correct the same elements in both $B$ and $C$ downward. Thus the sequences of $B_t$ and $C_t$ are also monotonically decreasing. Note that we would never want to revise these corrections back upward in a later iteration as these edges or caps will never be activated by liabilities that are lower element-wise than we have already tried in the previous round.
	
	In the equilibrium, all edge and cap activations will be supported by the equilibrium liability values. Eventually, we reach a state that represents the correct edge and cap activations, and the edge liabilities are the solution to the resulting linear system. This equilibrium point is the unique fixed point since solving the linear system and checking that $B$ and $C$ do not change is equivalent to verifying a fixed point of $\Phi$.
	
	At each step, either $B$ or $C$ changes or we stop our iteration. Thus there are at most $2m$ steps as there are at most $2m$ possible changes to $B$ and $C$. The most complex task at each step is again solving a linear system, which requires in general $O(m^3)$ time. Note that the $\Psi$ transformations are sparse (at most a single entry per row and column) and can be computed in at most $O(m^2)$. Thus the algorithm converges in at most $O(m^4)$ time.
\end{proof}

\section{Simulation Details}\label{appendix:simulations}

\subsection{Network Construction}

As the basis for our simulations, we use real network data on property and casualty reinsurance from 2012 Schedule F Part 3, as obtained from the National Association of Insurance Commissioners \citep{naic_reins_data}. This data details premiums ceded to reinsurers by US insurance companies. Naturally, this data does not provide all contract parameters, so we estimate these using common rules of thumb in the insurance industry, which we back up with data where available.

\subsubsection{XL contract parameters}

We construct networks of XL contracts consistent with the NAIC data by estimating the coverage provided by each firm's reinsurance contracts and separating its reinsurers into two layers. We introduce the following `in-the-ballpark' example of a reinsurance contract.\footnote{Private conversations with an insurance industry executive. All errors are our own.}

\begin{example}
	(`Ballpark' Reinsurance Contract) Suppose \$500M is the 1-in-100 year loss for a firm. As an `in-the-ballpark' figure, this firm would purchase reinsurance coverage of \$500M in losses with a deductible of \$100M. The \$400M total coverage limit would be separated equally into 2-3 layers. The total premiums ceded for this coverage would be 10\% of the \$400M limit. The lower layers would receive closer to 20\% of their respective limits, while the higher layer would receive closer to 2-3\% of its respective limit.
\end{example}

This example suggests the following rules of thumb that we use to fill in parameters in our real world network:
\begin{itemize}
	\item $\text{premiums ceded} \approx 0.1 \cdot \text{coverage limit}$,
	
	\item $\text{coverage limit} \approx 4 \cdot \text{deductible}$,
	
	\item $\text{coverage} \approx 5 \cdot \text{deductible}$, where $\text{coverage} = \text{coverage limit} + \text{deductible}$,
	
	\item $\text{top layer premiums} \approx 0.2 \cdot \text{total premiums ceded}$.
\end{itemize}
The only publicly available reinsurance contract data that we are aware of comes from major state catastrophe funds--for instance, the Florida Hurricane Catastrophe Fund and the Texas Windstorm Insurance Association. We compiled data on these reinsurance contracts, which is available in our code repository. This data supports that the first rule of thumb is reasonable.

Given a separation of a firm's reinsurers into layers, these rules of thumb allow us to estimate each contract's deductible and cap. We then estimate each contract's proportion of the layer as
$$\gamma = \text{premiums ceded}/\text{total premiums ceded for layer}.$$ Note that the coverage limits discussed above, which represent the cap payout from the whole reinsurance tower, are different from individual contract caps, which dictate the maximum payout from each contract that is itself only a part of the whole tower.

To separate a firm's reinsurers into two layers, we use the last rule of thumb to note that the premiums from the bottom layer should add to 80\% of total premiums and the premiums from the top layer should comprise the remaining 20\%. This is a knapsack problem that we can efficiently solve approximately.

\subsubsection{Proportional contract parameters}

We construct networks of proportional contracts consistent with the NAIC data by setting all contract deductibles to 0, all contract caps to $\infty$, and calculating each contract's coinsurance rate as
$$\gamma = \text{premium ceded} / (\text{primary premiums} + \text{foreign reinsurance premiums} + \text{reinsurance premiums}),$$
where the denominator describes insurance premiums received by the ceding firm in the contract. In this way, the ceding firm cedes a proportion of their total risk for the same portion of the premiums they have received. We estimate the primary premiums and foreign reinsurance premiums next.

\subsubsection{Primary insurance and foreign reinsurance premiums}

We additionally need to estimate the insurance premiums received from outside the reinsurance network. If the receiving firm is a primary insurer, this is the primary insurance premiums they receive. If the firm is a reinsurer, this is foreign (outside US) reinsurance premiums.

We then generate figures for these values within an estimated range:
\begin{enumerate}
	\item We collect data on premiums received and reinsurance premiums ceded from 10-ks and annual reports. Our data is available in our code repository. From this data, we determine reasonable upper and lower bounds on $\text{premiums ceded}/\text{premiums received}$ for firms.
	
	\item For each firm in the network, we generate a random number uniformly between the upper and lower bounds. This is used as the firm's $\text{premiums ceded}/\text{premiums received}$ ratio.
	
	\item From this ratio, we calculate the outside premiums--either primary insurance or foreign reinsurance--that the firm must receive to achieve this ratio. If the amount is negative, it is treated as zero.
\end{enumerate}

Based on the data, primary insurers generally have ratios between 0.05 and 0.5, and reinsurers generally have ratios between 0.1 and 0.3. We use these bounds in our simulations.

\subsubsection{Firm capital levels}

We next need to estimate each firm's capital that is available for paying its liabilities (i.e., the firm's equity). Current capital regulations focus on various factors through Risk-Based Capital; however, past regulations focused on the simpler leverage ratio \citep{reins_leverage_ratio}. For simplicity, we use this latter measure as a benchmark in our simulations. The leverage ratio is defined in the following way:
$$\text{leverage ratio} = \text{equity} / \text{net written premiums}.$$
According to \citep{reins_leverage_ratio_data}, American regulation required minimum 50\% leverage ratios. They also state that 20\% leverage ratios was a ``rule of thumb" in the German market for property and casualty insurers.

We extend this information by collecting data on equities and net written premiums from 10-ks and annual reports. Our data is available in our code repository. We use this data to determine reasonable upper and lower bounds on current leverage ratios. Based on the data, insurers generally have leverage ratios between 0.7 and 2.0, which we use as bounds in our simulations. We then generate figures for firm leverage ratios within the estimated range:

\begin{enumerate}
	\item For each firm in the network, we generate a random number uniformly between the upper and lower bounds. This is used as the firm's leverage ratio.
	
	\item From the ratio, we calculate the firm's equity.
\end{enumerate}

Note that following a market collapse, leverage ratios can plummet, which can significantly affect the capital levels in the reinsurance network. This is the reason that Risk-Based Capital is now used for regulation instead of the leverage ratio. For the price of adding greater complexity to our simulations, we could alternatively use Risk-Based Capital measures to estimate equity values instead.

\subsubsection{Shocks to primary insurers}

The final component of our simulation setup is to calibrate network shocks. These shocks are claims on primary insurers in the network. For our simulations, we consider 1-in-100 year and 1-in-250 year shocks. Industry data on the estimated aggregate size of these shocks is available from \citep{reins_tail_losses_data}. In particular, the North American 1-in-100 year insured loss is estimated at \$215.2B, and the North American 1-in-250 year insured loss is estimated at \$290.6B. We use these numbers for the aggregate size of tail shocks in our simulations.

The remaining task is to distribute this aggregate shock to primary insurers in the network. We do this in the following way:
\begin{enumerate}
	\item For each firm, we generate a random number uniformly between 0 and the size of that firm, defined by the total primary premiums received. Under this scheme, the size of a primary insurer correlates with their exposure. Reinsurers' initial exposure is zero as they do not offer primary insurance coverage.
	
	\item We then generate the shock exposure ratios by normalizing these numbers so that they add to 1. Multiplying by the size of the aggregate shock then gives the size of claims to primary insurers under the shock.
\end{enumerate}

In reality, the relation between the size of a primary insurer and its exposure to aggregate shocks is more complex that we model here. On one hand, larger primary insurers may be in a better position to diversify their holdings against geographic risk. On the other hand, their exposures could be higher since their portfolios are larger. In a more realistic model, we would want to account for the geographic exposures of each primary insurer and simulate geographic tail events. However, the data needed for this is not, in general, publicly available.

\subsection{Sensitivity to parameter perturbations}

In these simulations, we construct an XL network from our data as described in the previous subsections. This is our base case for comparison. In this process, we store the layering structure for future access. Given a multiplicative error term $\delta$ (i.e., a percentage error), we then perturb the network parameters as follows:
\begin{itemize}
	\item For each premium ceded value, we generate a random number uniformly in $[1-\delta,1+\delta]$ and multiply it with the premium ceded value.
	
	\item We then construct the contract parameters using the perturbed premium ceded values as described in the previous subsections using the stored layering structure.
	
	\item For each value of a firm's primary insurance premiums received, foreign reinsurance premiums received, and capital levels, we perturb it by a random multiplicative value uniformly chosen in $[1-\delta, 1+\delta]$.
\end{itemize}

\subsection{Systemic effects of contract structures}

We construct systems that are comparable given the structure of the graph on premiums ceded between insurers and our rules of thumb for XL reinsurance contracts. We construct comparable systems using the methods from the previous subsections keeping premiums ceded, firm capital levels, primary insurance premiums, and foreign reinsurance premiums constant. The only differences are in how the ceded premiums are interpreted: as part of a proportional scheme or XL contracts based on our rules of thumb and knapsack separation of layers.

\end{document}